\documentclass[prd,,aps,showpacs,nofootinbib,superscriptaddress,amsmath,amssymb]{revtex4}
\usepackage{amssymb,epsfig}
\usepackage{mathrsfs}
\usepackage[usenames,dvipsnames]{color}
\usepackage{multirow, array} 
\usepackage{float} 
\usepackage{amsthm}

\newcommand{\Real}{\mathbb{R}}

\newtheorem{Teo}{Theorem}

\newtheorem{Lem}{Lemma}

\begin{document}


\title{Static spherical perfect fluid stars with finite radius in general relativity: a review}

\author{Emmanuel Ch\'avez Nambo}
\affiliation{Instituto de F\'isica y Matem\'aticas, Universidad Michoacana de San Nicol\'as de Hidalgo,
Edificio C-3, Ciudad Universitaria, 58040 Morelia, Michoac\'an, M\'exico}

\author{Olivier Sarbach}
\affiliation{Instituto de F\'isica y Matem\'aticas, Universidad Michoacana de San Nicol\'as de Hidalgo,
Edificio C-3, Ciudad Universitaria, 58040 Morelia, Michoac\'an, M\'exico}

\date{\today}


\begin{abstract}
In this article, we provide a pedagogical review of the Tolman-Oppenheimer-Volkoff (TOV) equation and its solutions which describe static, spherically symmetric perfect fluid stars in general relativity. Our discussion starts with a systematic derivation of the TOV equation from the Einstein field equations and the relativistic Euler equations. Next, we give a proof for the existence and uniqueness of solutions of the TOV equation describing a star of finite radius, assuming suitable conditions on the equation of state characterizing the matter. We also prove that the compactness of the matter contained inside a sphere centered at the origin satisfies the well-known Buchdahl bound, independent of the radius of the sphere. Further, we derive the equation of state for an ideal, classical monoatomic relativistic gas from statistical mechanics considerations and show that it satisfies our assumptions for the existence of a unique solution describing a finite radius star. Although none of the results discussed in this article are new, they are usually scattered in different articles and books in the literature; hence it is our hope that this article will provide a self-contained and useful introduction to the topic of relativistic stellar models.
\end{abstract}



\maketitle

\section{Introduction}
\label{Sec:Intro}

The simplest model for describing a spherical star in equilibrium is the well-known Lane-Emden equation (see Ref.~\cite{chandrasekhar1957introduction} and references therein)
\begin{equation}
\frac{1}{x^2}\frac{d}{dx}\left( x^2 \frac{d\Theta}{dx}\right)  + \Theta^N = 0,
\label{Eq:LaneEmden}
\end{equation}
where $x$ represents a dimensionless radius, $\Theta^N$ is proportional to the mass density $\rho$ and $N$ is the polytropic index characterizing the equation of state of the matter. This model is based on the assumption of a static and spherically symmetric Newtonian perfect fluid with a polytropic equation of state in which the pressure $p$ is related to the density through the relation $p(\rho) = K\rho^{\gamma}$, $K$ being a constant and $\gamma = 1 + \frac{1}{N}$ the adiabatic index. Under these assumptions, Eq.~(\ref{Eq:LaneEmden}) easily follows from the condition of hydrostatic equilibrium and the Poisson equation for the gravitational potential, and it yields a simple and successful model that is able to describe (in first approximation) most of the stars in the Universe and even other astrophysical objects like planets. For example, our Sun can be described in first approximation by the Lane-Emden equation~(\ref{Eq:LaneEmden}) with polytropic index $N = 3$ ($\gamma = 4/3$), while low-mass white-dwarfs stars can be described by Eq.~(\ref{Eq:LaneEmden}) with index $N = 3/2$ ($\gamma = 5/3$). Giant planets, like Jupiter and Saturn, can be approximated by $N = 1$ ($\gamma = 2$) while the solution with $N = 0$ ($\gamma = \infty$) corresponds to a constant density, incompressible sphere and therefore serves as a simple model for rocky planets~\cite{chandrasekhar1957introduction,kTrB-Book}.

Although the Lane-Emden equation~(\ref{Eq:LaneEmden}) provides a simple model for most stars in our Universe, a more realistic description clearly requires additional physical ingredients, such as incorporating the effects of the rotation of the star, the presence of magnetic fields, radiation processes etc. Furthermore, if the star is very compact, then general relativistic effects become important. For a star of radius $R$ and mass $M$, the compactness is measured by the ratio $r_s/R$ where $r_s := 2G_N M/c^2$ is the Schwarzschild radius of the star (with $G_N$ and $c$ denoting, respectively, Newton's constant and the speed of light). More generally, the compactness ratio at radius $r$ is defined as $2m(r)/r$ with $m(r) := G_N M(r)/c^2$, where $M(r)$ denotes the mass contained in the sphere of radius $r$ centered at the origin. Relativistic corrections must be taken into account whenever this ratio ceases to be much smaller than one. This is the case for neutron stars or more exotic stars, like quark stars (see Refs.~\cite{Shapiro,Glendenning:1997wn} for textbooks treating these subjects).

In this article, we discuss the general relativistic generalization of the Lane-Emden equation, which is known as the Tolman-Oppenheimer-Volkoff (TOV) equation~\cite{rT39,jOgV39} and serves as a model for describing such compact stars, assuming they can still be modeled by a static and spherically symmetric perfect fluid. The TOV equation is obtained by replacing the Newtonian Euler-Poisson system by its relativistic generalization, the Euler-Einstein system of equations in which the self-gravity of the matter is described according to Einstein's theory of general relativity. This leads to generalizations of the hydrostatic equilibrium condition and Poisson's equations which correctly take into account the effects from general relativity and enhance the magnitude of the pressure gradient.

While a detailed mathematical analysis of the Lane-Emden equation~(\ref{Eq:LaneEmden}) has been known for a long time (see again~\cite{chandrasekhar1957introduction} and references therein and also~\cite{uS00} for the case of more general equations of state), a rigorous analysis of its relativistic counterpart has been completed only in more recent years. Pioneering work in this direction has started with the work by Rendall and Schmidt~\cite{Rendall_1991}, where it is shown that under certain assumptions on the equation of state, there exists for each value of the central density a unique global solution of the TOV equation in which the corresponding star either has a finite radius (and the solution being the Schwarzschild solution in the exterior region) or has infinite radius with the energy density converging to zero as $r\to \infty$.\footnote{Stars with infinite radius are relevant as well, as long as their density decays sufficiently fast to zero when $r\to \infty$ such that their total mass is finite. In particular, this is the case for boson stars, where the perfect fluid source of matter is replaced by a massive scalar field, see Ref.~\cite{sLcP17} for a recent review. For a recent study regarding the asymptotic behavior of some perfect fluid star models with infinite extend, see Ref.~\cite{lAaB19}.} Some necessary and sufficient conditions on the equation of state yielding a star with finite radius are also given in~\cite{Rendall_1991}. A different proof for the existence of solutions describing a star with finite radius was given by Makino~\cite{makino1998}, under assumptions on the equation of state which are similar to the one formulated in the next section of the present article, with the effective adiabatic index $\gamma$ being restricted to the range $4/3 < \gamma < 2$ for sufficiently small values of the density. The work in~\cite{makino1998} also discusses the radial linearized perturbations of the static solutions, showing that they lead to a self-adjoint operator with a purely discrete spectrum. For further work providing conditions on the equation of state which yield a spherical star of finite (or infinite) extend see Refs.~\cite{wS02,jH02}. In particular, the work by Simon~\cite{wS02} discusses the relation of these conditions with the uniqueness property of the static spherical stars among all possible static, asymptotically flat solutions of the Euler-Einstein equations. Other conditions that guarantee the finiteness of the star's radius have been presented by Ramming and Rein~\cite{Ramming_2013}. These conditions cover perfect fluid stars as well as self-gravitating collisionless gas configurations in both the Newtonian and relativistic regimes. For a general study of the relativistic spherically symmetric static perfect fluid models based on the theory of dynamical systems, see Ref.~\cite{Heinzle_2003}.

Coming back to the compactness ratio of the star (which determines when the relativistic effects are important), Buchdahl showed~\cite{Buchdahl:1959zz} that if the pressure is isotropic and the energy density does not increase outwards, then any static, spherically symmetric relativistic star must satisfy the inequality $r_s/R < 8/9$. This inequalities was later generalized by Andr\'easson~\cite{Andreasson:2007ck} who provides an $r$-independent bound on the compactness ratio $2m(r)/r$ under purely algebraic inequalities on the energy density and pressure, and hence removes the monotonicity assumption on the density profile.

The goal of this article is to provide a self-contained pedagogical review of the most important aspects of the TOV equation and its solutions. We start in section~\ref{Sec:TOV} with a systematic deduction of the TOV equation from the Euler-Einstein system of equations with a static, spherical ansatz, and we specify our assumptions on the equation of state. Moreover, in order to facilitate the mathematical analysis that follows, we rewrite the TOV equation in terms of dimensionless quantities. Next, in section~\ref{Sec:LocalExistence} we use the contraction mapping principle in order to prove the existence of a unique local solution for the dimensionless TOV equation near the center of symmetry $r = 0$. It should be noted that this step does not follow in a straightforward way from the standard results of the theory of ordinary differential equations, since the TOV equation is nonlinear and singular at $r = 0$. Next, in section~\ref{Sec:GlobalExistence} we prove that under the assumptions on the equation of state given in section~\ref{Sec:TOV} the local solution can be extended to either infinite radius or to a finite radius, and partly following~\cite{Ramming_2013} we prove that as long as the effective adiabatic index $\gamma$ is strictly larger than $4/3$ for small densities, the radius must be finite. Our proof also shows that the Buchdahl inequality $2m(r)/r < 8/9$ must hold for all values of the radius $r > 0$. A numerical example is analyzed in section~\ref{Sec:Numerical} and a summary and conclusions are presented in section~\ref{Sec:Conclusions}. This article also contains several appendices which provide technical details and some important examples. In appendix~{\ref{App:Curvature} we give details on the computation of the Riemann, Einstein and Ricci tensors which are used to derive the TOV equation. In appendix~\ref{App:StatFis} we provide a derivation of the equation of state describing a relativistic, ideal classical monoatomic gas from purely statistical physics considerations and mention the corresponding results for a completely degenerate ideal Fermi gas. In appendix~\ref{App:ModifiedBessel} we discuss some important properties of the modified Bessel functions of the second kind which are needed in appendix~\ref{App:StatFis}. In the final appendix~\ref{App:Completeness} we prove the completeness of the function space $X_R$ which plays a fundamental role for the local existence proof in section~\ref{Sec:LocalExistence}.

In most of the article, we work in geometrized units, for which $G_N = c = 1$.

\section{Derivation of the TOV equation and assumptions on the equation of state}
\label{Sec:TOV}

In this section, we start with a review of the derivation of the TOV equation. Then, we state the precise assumptions on the equation of state on which the results in the subsequent sections are based on.

\subsection{Field equations and static, spherically symmetric ansatz}

The field equations describing a relativistic, self-gravitating perfect fluid configuration are given by the coupled system consisting of the $10$ independent components of Einstein's field equations,
\begin{equation}
\label{Eq:Einstein}
G_{\mu\nu} = \frac{8\pi G_N}{c^4}T_{\mu\nu},
\end{equation}
together with the $4$ relativistic Euler equations
\begin{equation}
\label{Eq:Euler}
\nabla^\mu T_{\mu\nu} = 0.
\end{equation}
Here and in the following, Greek indices $\mu,\nu,\ldots$ denote spacetime indices which run over $0,1,2,3$, $G_{\mu\nu}$ are the components of the Einstein tensor associated with the spacetime metric $g_{\mu\nu}$ (which is symmetric, i.e. $G_{\mu\nu} = G_{\nu\mu}$ and hence has $10$ independent components like the metric components $g_{\mu\nu}$), and $T_{\mu\nu} = T_{\nu\mu}$ are the components of the energy-momentum-stress tensor which describes the sources of energy and matter. For the perfect fluid case considered here,
\begin{equation}
\label{Eq:EMST}
T_{\mu\nu} =  \frac{\varepsilon + p}{c^2} u_\mu u_\nu + p g_{\mu\nu},
\end{equation}
where $\varepsilon$, $p$ and $u^\mu = g^{\mu\nu} u_\nu$ refer, respectively, to the energy density, pressure and the components of the four-velocity of the fluid, normalized such that $u_\mu u^\mu = -c^2$. In terms of an orthonormal frame ${\bf e}_{\hat{0}},{\bf e}_{\hat{1}},{\bf e}_{\hat{2}},{\bf e}_{\hat{3}}$ of vector fields such that ${\bf e}_{\hat{0}} = c^{-1} u^\mu\partial_\mu$, the components of the energy-momentum-stress tensor are
\begin{equation}
\label{Eq:ComT}
(T_{\hat{\alpha}\hat{\beta}}) = \mbox{diag}(\varepsilon,p,p,p),
\end{equation}
and thus $\varepsilon$ and $p$ represent the energy density and pressure measured by an observer which is co-moving with the fluid (i.e. an observer whose world line is tangent to the four-velocity).

The Einstein tensor $G_{\mu\nu}$ is obtained from the Riemann curvature tensor $R^\alpha{}_{\beta\mu\nu}$ as follows:
\begin{equation}
G_{\mu\nu} = R_{\mu\nu} - \frac{R}{2}g_{\mu\nu},
\end{equation}
where $R_{\mu\nu} = R^{\alpha}{}_{\mu\alpha\nu}$ are the components of the Ricci tensor and its trace $R = g^{\mu\nu}R_{\mu\nu}$ is the Ricci scalar. The components of the Riemann curvature tensor, in turn, are given by
\begin{equation}
\label{Eq:Riemann}
R^\mu{}_{\nu\alpha\beta} = \partial_\alpha\Gamma^\mu{}_{\beta\nu} + \Gamma^\sigma{}_{\beta\nu}\Gamma^\mu{}_{\alpha\sigma} - (\alpha \leftrightarrow \beta) 
 = -R^\mu{}_{\nu\beta\alpha},
\end{equation}
where $\Gamma^{\nu}{}_{\alpha\beta}$ denote the Christoffel symbols, which are determined by the components of the metric tensor and their first derivatives,
\begin{equation}
\label{Eq:Christoffel}
\Gamma^\nu{}_{\alpha\beta} = \frac{1}{2}g^{\nu\sigma}\left(\frac{\partial g_{\beta\sigma}}{\partial x^{\alpha}} + \frac{\partial g_{\alpha\sigma}}{\partial x^{\beta}}  - \frac{\partial g_{\alpha\beta}}{\partial x^{\sigma}}\right).
\end{equation}
Due to the contracted Bianchi identities, $\nabla^\mu G_{\mu\nu} = 0$, Eq.~(\ref{Eq:Euler}) is a consequence of Einstein's field equations~(\ref{Eq:Einstein}), so in principle it is sufficient to solve Eq.~(\ref{Eq:Einstein}). However, as we will see, it is simpler to solve instead the relativistic Euler equations~(\ref{Eq:Euler}) together with part of the components of the Einstein equations.

For the remainder of this article, we focus on spherically symmetric and static configurations, in which the metric has the form
\begin{equation}
\label{Eq:MetricAnsatz}
ds^2 = g_{\mu\nu} dx^\mu dx^\nu
 = -e^{\frac{2\Phi(r)}{c^2}}c^2dt^2 + e^{2\Psi(r)}dr^2 
 + r^2(d\vartheta^2 +  \sin^2\vartheta d\varphi^2),
\end{equation}
where $(x^\mu) = (t, r, \vartheta, \varphi)$ are spherical coordinates and $\Phi$ and $\Psi$ are functions of the radius coordinate $r$ only which will be determined by the field equations~(\ref{Eq:Einstein},\ref{Eq:Euler}). Note that when $\Phi = \Psi = 0$, the metric~(\ref{Eq:MetricAnsatz}) reduces to the Minkowski metric in spherical coordinates. In the solutions discussed below, the coordinate $r$ runs from $0$ to $\infty$. For the solution to be regular at $r=0$ we require $\Phi(r)$ and $\Psi(r)$ to be smooth, even functions of $r$ (i.e. all their derivatives of odd order vanish at $r=0$). As $r\to \infty$ we require asymptotic flatness, that is $\Phi,\Psi \to 0$. The perfect fluid configuration is also assumed to be static and spherically symmetric. This means that $\varepsilon = \varepsilon(r)$ and $p = p(r)$ are functions of $r$ only, and that the four-velocity is of the form
\begin{equation}
u^\mu\frac{\partial}{\partial x^\mu} = e^{-\frac{\Phi}{c^2}}\frac{\partial}{\partial t},
\end{equation}
such that the fluid elements are at rest in the reference frame defined by the coordinate system $(t,r,\vartheta,\varphi)$.

\subsection{Explicit expressions for the Einstein tensor and exterior solution}

In order to compute the $10$ independent components of the Einstein tensor $G_{\mu\nu}$ appearing in Eq.~(\ref{Eq:Einstein}), one needs to calculate first the $40$ independent Christoffel symbols $\Gamma^\nu{}_{\alpha\beta}$, as explained in the previous subsection. To carry out this calculation, it is convenient to exploit the block-diagonal form of the metric and write it as follows:
\begin{equation}
(g_{\mu\nu}) = \begin{pmatrix} \tilde{g}_{ab} & 0 \\ 0 & r^2\hat{g}_{AB} \end{pmatrix}, \qquad (g^{\mu\nu}) = \begin{pmatrix} \tilde{g}^{ab} & 0 \\ 0 & r^{-2}\hat{g}^{AB} \end{pmatrix},
\label{Eq:SphMetric}
\end{equation}
where $a,b$ refer to the coordinates $t,r$ and $A,B$ to the coordinates $\vartheta,\varphi$. For the specific parametrization~(\ref{Eq:MetricAnsatz}) relevant to this section,  the two blocks are given by
\begin{align}
\tilde{g}_{ab}dx^adx^b &= -e^{2\Phi(r)}dt^2 + e^{2\Psi(r)}dr^2, 
& (\hbox{$a, b = t, r$}), \label{Eq:TwoMetric}\\
\hat{g}_{AB}dx^Adx^B &= d\vartheta^2 +  \sin^2\vartheta d\varphi^2, 
& (\hbox{$A, B = \vartheta, \varphi$}). \label{Eq:SphTwoMetric}
\end{align}
From now on, we work in geometrized units in which $G_N = c = 1$, implying in particular that time and mass have units of length. The details of the calculations are presented in Appendix~\ref{App:Curvature}; here we directly present the resulting expressions for the Christoffel symbols and the components of the Einstein tensor. The non-vanishing Christoffel symbols are:
\begin{eqnarray}
&& \Gamma^{t}{}_{tr} = \Gamma^{t}{}_{rt} = \Phi', \qquad
\Gamma^{r}{}_{rr} = \Psi', \qquad
\Gamma^{r}{}_{tt}  = \Phi' e^{2(\Phi - \Psi)},
\label{Eq:Christoffel1}\\
&& \Gamma^{\vartheta}{}_{r\vartheta} =  \Gamma^{\vartheta}{}_{\vartheta r}
 = \Gamma^{\varphi}{}_{r\varphi} = \Gamma^{\varphi}{}_{\varphi r} = \frac{1}{r},
\label{Eq:Christoffel2}\\
&& \Gamma^{r}{}_{\vartheta\vartheta} = -re^{-2\Psi}, \qquad
\Gamma^{r}{}_{\varphi\varphi} = -r\sin^2\vartheta e^{-2\Psi},
\label{Eq:Christoffel3}\\
&& \Gamma^{\vartheta}{}_{\varphi\varphi}  = -\sin\vartheta\cos\vartheta, \qquad
\Gamma^{\varphi}{}_{\varphi\vartheta} = \Gamma^{\varphi}{}_{\vartheta\varphi} = \cot\vartheta,
\label{Eq:Christoffel4}
\end{eqnarray}
which give rise to the following expressions for the Einstein tensor:
\begin{align}
\label{Eq:TE1}
G^{t}{}_{t} & = \frac{1}{r^2}\left(e^{-2\Psi} - 1\right) - \frac{2\Psi'}{r}e^{-2\Psi}, \\
\label{Eq:TE2}
G^{r}{}_{r} & = \frac{1}{r^2}\left(e^{-2\Psi} - 1\right) + \frac{2\Phi'}{r}e^{-2\Psi}, \\
\label{Eq:TE3}
G^{\vartheta}{}_{\vartheta} = G^{\varphi}{}_{\varphi} & = \left[\Phi'' + \Phi'(\Phi' - \Psi') + \frac{\Phi' - \Psi'}{r}\right]e^{-2\Psi},
\end{align}
the off-diagonal components being zero.

Based on these expressions, it is a simple task to derive the Schwarzschild metric, which describes the unique static, spherically symmetric family of solutions in the exterior vacuum region. In vacuum, there are no energy sources and thus $T_{\mu\nu} = 0$ and Einstein's field equations imply
\begin{align}
\label{Eq:S1}
\frac{1}{r^2}\left(e^{-2\Psi} - 1\right) - \frac{2\Psi'}{r}e^{-2\Psi} & = 0, \\
\label{Eq:S2}
\frac{1}{r^2}\left(e^{-2\Psi} - 1\right) + \frac{2\Phi'}{r}e^{-2\Psi} & = 0, \\
\label{Eq:TS3}
\left[\Phi'' + \Phi'(\Phi' - \Psi') + \frac{\Phi' - \Psi'}{r}\right]e^{-2\Psi} & = 0.
\end{align}
The first equation only involves $\Psi(r)$ and can be rewritten as
\begin{equation}
G^t{}_t = -\frac{1}{r^2}\frac{d}{dr}[r(1 - e^{-2\Psi})] = 0,
\end{equation}
and hence $r(1 - e^{-2\Psi}) = 2M$ for some integration constant $M$. For reasons which will become clear shortly, we assume $M > 0$ to be positive. Therefore,
\begin{equation}
\label{Eq:aSch}
e^{-2\Psi} = 1 - \frac{2M}{r}.
\end{equation}
Moreover, subtracting Eq.~(\ref{Eq:S1}) from (\ref{Eq:S2}) one obtains the relation
\begin{equation}
\Phi' = -\Psi',
\end{equation}
which can be integrated to yield
\begin{equation}
\Phi = -\Psi,
\end{equation}
where without loss of generality we have set the integration constant to zero, since otherwise it could be absorbed into a redefinition of the time coordinate $t$ (which does not alter the physics of the problem because of the general covariance principle of General Relativity). Using this relation in Eq.~(\ref{Eq:aSch}) one obtains 
\begin{equation}
\label{Eq:PhiSch}
e^{2\Phi} = 1 - \frac{2M}{r},
\end{equation}
which yields the Schwarzschild solution, given by the line element
\begin{equation}
ds^2 = -\left(1 - \frac{2M}{r}\right)dt^{2} + \left(1 - \frac{2M}{r}\right)^{-1}dr^{2} + r^2(d\vartheta^{2} + \sin^{2}\vartheta d\varphi^{2}). 
\end{equation}
We see that for $r \gg M$, $2M/r \ll 1$, and in this limit the metric can be considered to describe a small perturbation of the flat Minkowski metric. Thus, in this case the Newtonian limit is valid which allows one to identify the quantity $-M/r$ with the Newtonian potential $\Phi$, that is, $\Phi = -M/r$. In this sense, the integration constant $M$ can be identified with the total mass of the central object. The Schwarzschild metric is an exact non-trivial (i.e. non-flat) solution of the Einstein field equations. In the absence of matter, it describes a non-rotating black hole (see, for instance, Ref.~\cite{Wald} for details).

\subsection{Interior region and TOV equations}

In the interior region, the relevant field equations are obtained by replacing the right-hand sides of Eqs. (\ref{Eq:S1})-(\ref{Eq:TS3}) with the corresponding components of $8\pi$ times the energy-momentum-stress tensor.\footnote{Recall that we work in geometrized units in which $G_N = c = 1$.} Using the fact that $T^t{}_t = -\varepsilon$, $T^r{}_r = T^\vartheta{}_\vartheta = T^\varphi{}_\varphi = p$, we obtain the following three equations
\begin{align}
\label{Eq:Einstein1}
\frac{1}{r^2}\left(e^{-2\Psi} - 1\right) - \frac{2\Psi'}{r}e^{-2\Psi} & = -8\pi\varepsilon, \\
\label{Eq:Einstein2}
\frac{1}{r^2}\left(e^{-2\Psi} - 1\right) + \frac{2\Phi'}{r}e^{-2\Psi} & = 8\pi p, \\
\label{Eq:Einstein3}
\left[\Phi'' + \Phi'(\Phi' - \Psi') + \frac{\Phi' - \Psi'}{r}\right]e^{-2\Psi} & = 8\pi p.
\end{align}
As in the vacuum case, the left-hand side of Eq.~(\ref{Eq:Einstein1}) only involves the metric field $\Psi(r)$, and it can be rewritten in the form
\begin{equation}
\frac{1}{r^2}\frac{d}{dr}[r(1 - e^{-2\Psi})] = 8\pi\varepsilon.
\end{equation}
Integrating both sides of this equation yields
\begin{equation}
\label{Eq:aTOV}
e^{-2\Psi(r)} = 1 - \frac{8\pi}{r} \int_{0}^{r} \varepsilon(s) s^2 ds,
\end{equation}
where we have used the fact that $\Psi(r)$ is regular at $r = 0$ to fix the integration constant. Introducing the mass function
\begin{equation}
\label{Eq:MasaT}
m(r) := 4\pi\int_{0}^{r} \varepsilon(s) s^2 ds,
\end{equation}
which measures the mass-energy contained in a sphere of radius $r$, Eq.~(\ref{Eq:aTOV}) can be rewritten as
\begin{equation}
\label{Eq:masa}
e^{-2\Psi(r)} = 1 - \frac{2m(r)}{r}.
\end{equation}
Eliminating the factor $e^{-2\Psi(r)}$ from Eq. (\ref{Eq:Einstein2}) one obtains 
\begin{equation}
\label{Eq:Phi}
\Phi'(r) = \frac{m(r) + 4\pi r^3 p(r)}{r[r - 2m(r)]}.
\end{equation}
This is the relativistic generalization of the Newtonian equation $\Phi'(r) = m(r)/r^2$, to which Eq.~(\ref{Eq:Phi}) reduces to in the limit $p \ll \varepsilon$ and $m(r) \ll r$.

Next, one needs an equation for the pressure $p(r)$. Such an equation could be obtained by substituting Eqs.~(\ref{Eq:aTOV}) and (\ref{Eq:Phi}) into the last Einstein equation~(\ref{Eq:Einstein3}). However, a lot of algebraic work can be saved by considering instead Eq.~(\ref{Eq:Euler}), from which one directly obtains the same result, which is
\begin{equation}
p' = -(p + \varepsilon)\Phi'.
\end{equation} 
Finally, we may eliminate $\Phi'$ from this equation by using Eq.~(\ref{Eq:Phi}), obtaining the well-known Tolman-Oppenheimer-Volkoff (TOV) equation
\begin{equation}
\label{Eq:TOV}
p'(r) = -[p(r) + \varepsilon(r)]\frac{m(r) + 4\pi r^3 p(r)}{r[r - 2m(r)]}.
\end{equation}
This generalizes the Newtonian condition for hydrostatic equilibrium $p'(r) = -\rho(r)\frac{m(r)}{r^2}$ (with $\rho$ the mass density) to the general relativistic case. Note that the relativistic correction terms tend to increase the pressure gradient $|p'|$, yielding more compact objects. Note also that Eq.~(\ref{Eq:TOV}) is singular at $r = 0$ and $2m(r) = r$.  The first one requires appropriate regularity conditions at the center and will be dealt with by replacing the mass function $m(r)$ with the mean density (see sections~\ref{SubSec:Dimensionless} and~\ref{Sec:LocalExistence} below). Regarding the potential singularity at $2m(r) = r$, we will prove in section~\ref{Sec:GlobalExistence} that (under the hypotheses made in this article), $2m(r) < r$ everywhere, such that it does not occur. For now we note that Eq.~(\ref{Eq:MasaT}) implies that $m(r)\simeq r^3$ near the center such that $2m(r)/r\simeq r^2$.

In summary, the metric for a spherical, static, self-gravitating perfect fluid configuration is given by
\begin{equation}
ds^2 = -e^{2\Phi(r)}dt^2 + \left(1 - \frac{2m(r)}{r}\right)^{-1}dr^2 + r^2(d\vartheta^2 + \sin^2\vartheta d\varphi^2),
\end{equation}
where $m(r)$ is given by Eq. (\ref{Eq:MasaT}), $\Phi(r)$ is determined from Eq. (\ref{Eq:Phi}), and $p(r)$ must satisfy the TOV equation~(\ref{Eq:TOV}). The latter can be integrated as soon as one specifies an equation of state which provides a relation between the pressure $p$ and the energy density $\varepsilon$. In the next subsection we specify our precise assumptions on the equations of state considered in this article, while in the subsequent sections we provide a rigorous analysis for the existence of solutions of the TOV equation.


\subsection{The equation of state}
\label{sec:EquationState}

In the following, we state our assumptions on the equation of state, which provides a relation between the pressure $p$ and the energy density $\varepsilon$. Such a relation should be obtained from a statistical mechanics model of the matter, which usually provides the pressure and energy density as a function of the particle density $n$ and the temperature $T$ of the system:
\begin{equation}
p = p(n,T),\qquad
\varepsilon = \varepsilon(n,T),
\end{equation}
see Appendix~\ref{App:StatFis} for the specific example of an ideal monoatomic relativistic gas. For the following, we assume that the perfect fluid configuration is in \emph{local thermodynamic equilibrium}, that is, each fluid (or gas) cell is in thermodynamic equilibrium and thus the macroscopic quantities describing the state of this cell satisfy the laws of thermodynamics. Assuming that the cell contains a fixed number $N$ of particles, the relevant macroscopic quantities characterizing the state of the cell are its volume $V = N/n$, its entropy $S = s N/n$ (with $s$ the entropy density), its energy $U = \varepsilon N/n$, and  other quantities such as its temperature $T$. Since $N$ is fixed, the first law of thermodynamics implies that
\begin{equation}
d\left(\frac{\varepsilon}{n}\right) = T d\left( \frac{s}{n} \right) -p d\left(\frac{1}{n}\right).
\label{Eq:FirstLaw}
\end{equation}
In general, the energy density $\varepsilon$ is a function of the entropy per particle $s/n$ and $n$; however, in this article we assume the perfect fluid is \emph{isentropic}, that is, $s/n$ is constant throughout the fluid, such that the first term on the right-hand side of Eq.~(\ref{Eq:FirstLaw}) can be ignored. In this case, $\varepsilon$ depends only on $n$ and given an equation of state in the form $p = p(n)$, integration of Eq.~(\ref{Eq:FirstLaw}) yields
\begin{equation}
\varepsilon(p) = ne_0 + n\int_{0}^{n} p(\overline{n})\frac{d\overline{n}}{\overline{n}^2},\qquad
p = p(n),
\label{Eq:epsilonp}
\end{equation}
where $e_0$ denotes the rest mass energy of the particle and where from now on, we regard $\varepsilon$ as a function of $p$ instead of $n$. More precisely, we assume $p: [0,\infty)\to \Real$ is a continuously differentiable function of the particle density $n$, satisfying the following conditions:
\begin{itemize}
\item[$(i)$] $p(n) > 0$ for $n > 0$ (positive pressure)
\item[$(ii)$] $p$ is monotonously increasing
\item[$(iii)$] Introducing the effective adiabatic index
\begin{equation}
\label{Eq:gamma(n)}
\gamma(n) := \frac{\partial\log p}{\partial\log n}(n) = \frac{n}{p(n)}\frac{\partial p}{\partial n}(n),
\qquad n > 0,
\end{equation}
we assume there is a constant $\gamma_1 > 1$ such that, for all small enough $n$,
\begin{equation}
\gamma(n)\geq \gamma_1
\end{equation}
\item[$(iv)$] $e_0 > 0$ (positive rest mass energy)
\end{itemize}

The condition $(iii)$ implies that for small enough $n_2\geq n_1 > 0$,
\begin{equation}
\frac{p(n_1)}{p(n_2)} \leq \left( \frac{n_1}{n_2} \right)^{\gamma_1},
\label{Eq:pInequality}
\end{equation}
which implies that $p(n)$ converges to zero at least as fast as $n^{\gamma_1}$ for $n\to 0$. In particular, this assures that the integral in Eq.~(\ref{Eq:epsilonp}) is well-defined, and it follows from the conditions $(i)$--$(iv)$ that $\varepsilon : [0, \infty)\to \Real$ is a continuously differentiable, monotonously increasing function which satisfies $\varepsilon(p)/n \to e_0$ as $p \to 0$.

For a discussion of realistic equations of state, including those describing phase transitions, we refer the reader to Ref.~\cite{Glendenning:1997wn}. In this case, the function $\varepsilon(p)$ might be discontinuous; however, it seems that models for neutron star matter based on two conserved quantities (baryonic number and electric charge) do yield a continuous relation between $n$, $p$ and $\varepsilon$, see chapter~9 in~\cite{Glendenning:1997wn}. See also Refs.~\cite{NeutronStarStructure,NuclearEquation,DenseMatter,MassesRadii} for recent work and reviews on realistic equations of state describing dense matter in neutron stars.


\subsection{Dimensionless field equations and summary}
\label{SubSec:Dimensionless}

For the analysis in the following sections it is useful to introduce the averaged energy density $\overline{\rho}(r)$ contained in a sphere of radius $r$:
\begin{equation}
\label{Eq:W}
\overline{\rho}(r) := \frac{m(r)}{\frac{4\pi}{3} r^3} 
 = \frac{3}{r^3}\int_0^r \varepsilon(p(s)) s^2 ds,\qquad r > 0,
\end{equation}
which is regular at the center. In terms of $\overline{\rho}(r)$, Eqs.~(\ref{Eq:Phi},\ref{Eq:TOV}) can be rewritten as
\begin{equation}
\label{Eq:TOV2}
\Phi'(r) = -\frac{p'(r)}{p(r) + \varepsilon(p(r))} 
 = \frac{4\pi r}{3} \frac{\overline{\rho}(r) + 3p(r)}{1 - \frac{8\pi}{3} r^2 \overline{\rho}(r)}.
\end{equation}
Furthermore, it is also very convenient for the following to work in terms of dimensionless quantities.  For this reason, we write the radius, pressure, energy density and averaged energy density as follows:
\begin{equation}
r  = \ell x,\qquad
p(r) = p_c P(x), \qquad
\varepsilon(p) = \varepsilon_c e(P), \qquad
\overline{\rho}(r) = \varepsilon_c w(x), 
\label{Eq:Dimensionless}
\end{equation}
where $p_c = p(0)$ is the central pressure, $\varepsilon_c$ the central energy density, and $\ell$ is a free parameter which will be chosen later. Here, the function $e(P)$ represents the dimensionless equation of state which satisfies the same properties as the function $\varepsilon(p)$ in Eq.~(\ref{Eq:epsilonp}). By definition, the functions $P(x)$, $w(x)$ and $e(P)$ satisfy the following conditions at the center,
\begin{equation}
P(0) = w(0) = 1, \qquad
e(1) = 1.
\label{Eq:CenterConditions}
\end{equation}
In terms of these quantities, the field equations~(\ref{Eq:TOV2}) are
\begin{equation}
\frac{d}{dx}\left(\frac{\Phi}{\lambda}\right) = -\frac{1}{e + \lambda P}\frac{dP}{dx} 
 = \frac{4\pi\ell^2 \varepsilon_c x}{3\lambda} \frac{w(x) 
  + 3\lambda P(x)}{1 - \frac{8\pi\ell^2\varepsilon_c}{3}x^2 w(x)},
\label{Eq:TOV3}
\end{equation}
where we have introduced the dimensionless parameter
\begin{equation}
\lambda := \frac{p_c}{\varepsilon_c},
\end{equation} 
representing the ratio between the central pressure and energy density. Note that in the Newtonian limit $\lambda \to 0$ since in this case the energy density and pressure are dominated by the contribution from the rest mass. In this sense, the parameter $\lambda$ measures how relativistic the resulting configuration will be. We see from Eq.~(\ref{Eq:TOV3}) that it is convenient to choose the length scale parameter $\ell$ such that
\begin{equation}
\frac{4\pi\ell^2 \varepsilon_c}{3} = \lambda.
\label{Eq:lDef}
\end{equation}
Also introducing the function $\phi(x) := \Phi(r)/\lambda$, our final form of the dimensionless field equations is
\begin{equation}
\label{Eq:TolmanA}
\frac{d}{dx}\phi(x) = 
 -\frac{1}{e(P(x)) + \lambda P(x)} \frac{d}{dx} P(x) 
  = x \frac{w(x) + 3\lambda P(x)}{1 - 2\lambda x^2 w(x)},
\end{equation}
with
\begin{equation}
\label{Eq:WA}
w(x) = \frac{3}{x^{3}}\int_0^x e(P(y)) y^2 dy.
\end{equation}
Note that in the Newtonian limit $\lambda\to 0$, Eq.~(\ref{Eq:TolmanA}) reduces to
\begin{equation}
\label{Eq:TolmanNewton}
\frac{d}{dx}\phi(x) = 
 -\frac{1}{e(P(x))} \frac{d}{dx} P(x) = x w(x),
\end{equation}
which are the correct Newtonian equations.

\section{Local existence near the center}
\label{Sec:LocalExistence}

In this section we prove, for each value $p_c > 0$ of the central pressure, the existence of a unique local solution $p(r)$ of the TOV equation~(\ref{Eq:TOV}) in the vicinity of the center of symmetry $r = 0$ such that $p(0) = p_c$. In the next section, this solution will be shown to possess a unique extension to a solution $p: [0,R_*]\to \Real$ of Eq.~(\ref{Eq:TOV}) which is monotonically decreasing and satisfies $p(R_*) = 0$, and hence describes a spherical static star of finite radius $R_*$.

In order to demonstrate the existence of the local solution of the TOV equation, we rewrite Eq.~(\ref{Eq:TolmanA}) as a fixed point problem and use the contraction mapping principle. For this, we integrate both sides of
\begin{equation}
\frac{d}{dx}P(x) = -[e(P(x)) + \lambda P(x)]x \frac{w(x) + 3\lambda P(x)}{1 - 2\lambda x^2 w(x)},
\label{Eq:TolmanABis}
\end{equation}
over $x$, obtaining (taking into account the central condition $P(0) = 1$ from Eq.~(\ref{Eq:CenterConditions})) the integral equation
\begin{equation}
\label{Eq:IntTOV}
P(x) = 1 - \int_0^x \left[e(P(y)) + \lambda P(y) \right] 
\frac{w(y) + 3\lambda P(y)}{1 - 2\lambda w(y) y^2} y dy =: TP(x),
\end{equation} 
where $w(x)$ is given by (\ref{Eq:WA}). The problem now consists in finding a function $P(x)$ (in a suitable function space which will be specified below) which satisfies the fixed point equation $P = TP$. This can be achieved by means of the contraction mapping principle, which provides sufficient conditions for $T$ to posses a unique fixed point. We recall this important result which can be found in many textbooks (see, for instance~\cite{ReedSimon80}).

\begin{Teo}[contraction mapping principle]
\label{Thm:Banach}
Let $\left(X, \|\cdot\|\right)$ be a Banach space, and let $A = \overline{A} \subset X$ be a closed, non-empty subset of $X$. Let $T : A \rightarrow A$ be a mapping from $A$ to itself which constitutes a contraction, that is, there exists a constant $L$ satisfying $0 \leq L < 1$ such that
\begin{equation}
\| T(u) - T(v)\| \leq L\|u - v\|\qquad \hbox{for all $u,v\in A$}.
\end{equation}
Then, $T$ has a unique fixed point $u^*\in A$, that is, there exists a unique $u^* \in A$ such that $T(u^*) = u^*$.\footnote{The theorem says even more: the unique fixed point $u^*\in A$ can be  obtained as the limit of the sequence $(u_k)$ defined by
\begin{equation*}
u_1 := T(u), \quad u_2 := T^2(u) = T(T(u)), \quad \ldots, \qquad u_k := T^k(u),
\end{equation*}
starting from any point $u\in A$. This sequence converges exponentially fast to $u^*$ as the following error bound shows:
\begin{equation*}
\|u_k - u^*\| \leq \frac{L^k}{1 - L}\|u_1 - u\|, \qquad k = 1, 2, 3, \ldots
\end{equation*}
}
\end{Teo}

In order to apply this theorem to the fixed point problem~(\ref{Eq:IntTOV}) we introduce, for each $R > 0$, the space $X_R := C_b( (0,R],\Real)$ of bounded, continuous real-valued functions on the interval $(0,R]$, equipped with the infinity norm:
\begin{equation}
\|P\|_{\infty} := \sup_{0 < x \leq R} |P(x)|, \quad P \in X_R.
\label{Eq:Norm}
\end{equation}
In Appendix~\ref{App:Completeness} we show that $\|\cdot\|_\infty$ defines a norm on $X_R$ and that $(X_R,\|\cdot\|_\infty)$ defines a Banach space, that is, a complete normed vector space. Next, we introduce the subset $A_R \subset X_R$ defined as
\begin{equation}
A_R := \left\{P \in X_R \; \biggr\rvert \; \lim\limits_{x\to 0} P(x) = 1 \; \hbox{and} \; \frac{1}{2} \leq P(x) \leq 1\; \hbox{for all} \; 0 < x \leq R\right\}.
\end{equation}
Clearly, $A_R$ is not empty since it contains the constant function $P = 1$. Furthermore, it is not difficult to verify that $A_R$ is closed: if $P_k$ is a sequence in $A_R$ which converges to $P\in X_R$ in the infinity norm, that is,
\begin{equation}
\| P_k - P \|_\infty = \sup_{0 < x \leq R} |P_k(x) - P(x)| \to 0,\qquad k\to \infty,
\end{equation}
then $P_k$ converges uniformly to $P$ and it follows that $P(x)\to 1$ as $x\to 0$ and $\frac{1}{2}\leq P(x)\leq 1$ since $P_k\in A_R$. Therefore, the limiting point $P$ of the sequence $P_k$ also lies in $A_R$, and it follows that $A_R$ is closed.

For the following, we show that the map $T$ defined in Eq.~(\ref{Eq:IntTOV}) is well-defined on $A_R$,  maps $A_R$ into itself and defines a contraction provided that $R > 0$ is small enough. For this, first note that due to the fact that $e(P)$ is an increasing function and that $P\leq 1$ it follows from Eq.~(\ref{Eq:WA}) and the normalization $e(1) = 1$ that
\begin{equation}
w(x) = \frac{3}{x^3}\int_0^x e(P(y)) y^2 dy \leq \frac{3}{x^3}\int_0^x e(1) y^2 dy = 1,
\end{equation}
for all $P\in A_R$, such that $w(x)$ is bounded from above by $1$. Also, since $P \geq 1/2$ for all $P\in A_R$, it follows that
\begin{equation}
w(x) = \frac{3}{x^3}\int_0^x e(P(y)) y^2 dy \geq \frac{3}{x^3}\int_0^x e(1/2) y^2 dy 
 = e(1/2) =: w_0 > 0,
\end{equation}
which allows us to conclude that $w_0 \leq w \leq 1$ for all $P \in A_R$. Moreover, since $e$ and $P$ are continuous, it follows that $w$ is continuous and (using L'H\^opital's rule) that $w(x)\to e(1) = 1$ as $x\to 0$. Thus, if the function $P$ lies in the set $A_R$, then the function $w$ defined by Eq.~(\ref{Eq:WA}) belongs to the set
\begin{equation}
B_R := \left\{w \in X_R \; \biggr\rvert \; \lim\limits_{x\to 0} w(x) = 1 \; \hbox{and} \; w_0 \leq w(x) \leq 1\; \hbox{for all} \; 0 < x \leq R\right\}.
\end{equation}

After these preliminary remarks, we are ready to show that the map $T$ in Eq.~(\ref{Eq:IntTOV}) defines a contraction on $A_R$, provided $R > 0$ is small enough: first, we observe that $1 - 2\lambda w(y) y^2 \geq 1 - 2\lambda R^2$ for all $0 < y \leq R$ if $w\in B_R$, such that the denominator in the integrand of Eq.~(\ref{Eq:IntTOV}) cannot vanish if $0 < x \leq R$ and $R$ is chosen small enough, such that $2\lambda R^2 < 1$. Next, using again the continuity and boundedness of the functions $e$, $P$ and $w$, it follows that $TP: (0,R]\to \Real$ is continuous and satisfies $TP(x)\to 1$ for $x\to 0$. Moreover, because the integrand in Eq.~(\ref{Eq:IntTOV}) is positive, it follows that $TP$ is monotonously decreasing. To show that $TP\in A_R$ it thus remains to prove that $TP(R) \geq \frac{1}{2}$. For this, we use the estimates $P\leq 1$, $w\leq 1$, $1 - 2\lambda w(y) y^2 \geq 1 - 2\lambda y^2$ and the fact that $e$ is an increasing function in order to estimate
\begin{equation*}
 [e(P(y)) + \lambda P(y)] \frac{w(y) + 3\lambda P(y)}{1 - 2\lambda w(y) y^2} 
\leq (1 + \lambda)\frac{1 + 3\lambda}{1 - 2\lambda y^2},
\end{equation*}
which implies
\begin{eqnarray*}
TP(x) &=&
1 - \int_0^x [e(P(y)) + \lambda P(y)] \frac{w(y) + 3\lambda P(y)}{1 - 2\lambda w(y) y^2} y dy \\
&\geq& 1 - \int_0^x (1 + \lambda)\frac{1 + 3\lambda}{1 - 2\lambda y^2} y dy \\
&=& 1 + \left(1 + \lambda \right)\frac{1 + 3\lambda}{4\lambda}\log(1 - 2\lambda x^2),
\end{eqnarray*}
for all $0 < x \leq R$, and the required condition $T P(R) \geq \frac{1}{2}$ is satisfied if $R > 0$ is small enough, such that
\begin{equation}
\label{Eq:De}
2\lambda R^2 \leq 1 - e^{-\frac{2\lambda}{(1+\lambda)(1 + 3\lambda)}},
\end{equation}
which is slightly stronger than the previous requirement $2\lambda R^2 < 1$. Therefore, if $R$ satisfies the inequality~(\ref{Eq:De}), the map $T$ defined by Eq.~(\ref{Eq:IntTOV}) is a well-defined map from $A_R$ into itself. To apply the contraction mapping principle, it remains to prove that $T$ defines a contraction on $A_R$ (for sufficiently small $R > 0$), that is, there must exist a constant $0 \leq L < 1$ such that
\begin{equation}
\|TP_2 - TP_1\|_{\infty} \leq L\|P_2 - P_1\|_{\infty}, \quad \hbox{for all $P_1, P_2 \in A_R$}.
\end{equation}
In order to verify this condition, we write the difference $TP_2 - TP_1$ in the following form:
\begin{equation}
TP_2(x) - TP_1(x) = -\int_0^x
\left[ F_{\lambda}(P_2(y), w_2(y), y) - F_{\lambda}(P_1(y), w_1(y), y) \right] y dy,
\label{Eq:TPDiff}
\end{equation}
with $F_\lambda : \left[\frac{1}{2}, 1\right] \times [w_0, 1] \times [0, R]\to \Real$ the continuously differentiable function defined by
\begin{equation}
F_\lambda(p, w, y) := [e(p) + \lambda p]\frac{w + 3\lambda p}{1 - 2\lambda w y^2},\qquad
\frac{1}{2}\leq p\leq 1,\quad w_0\leq w\leq 1,\quad 0\leq y\leq R.
\end{equation}
According to the mean value theorem \cite{Apostol}, one has for all $\frac{1}{2}\leq P_1,P_2\leq 1$, $w_0\leq w_1,w_2\leq 1$ and $0\leq y\leq R$,
\begin{equation}
F_{\lambda}(P_2, w_2, y) - F_{\lambda}(P_1, w_1, y) 
 = \frac{\partial F_{\lambda}}{\partial P}(P_*, w_*, y)(P_2 - P_1)
 + \frac{\partial F_{\lambda}}{\partial w}(P_*, w_*, y)(w_2 - w_1),
\end{equation}
with $P_* = P_1 + \theta_P(P_2 - P_1)$, $0 < \theta_P < 1$, lying between $P_1$ and $P_2$ and likewise, $w_* = w_1 + \theta_w(w_2 - w_1)$, $0 < \theta_w < 1$. Using this into Eq.~(\ref{Eq:TPDiff}) one obtains the estimate
\begin{eqnarray}
|TP_2(x) - TP_1(x)| &\leq& \int_0^x 
\left[ \left|\frac{\partial F_{\lambda}}{\partial P}(P_*(y), w_*(y), y)(P_2(y) - P_1(y))\right| 
+ \left|\frac{\partial F_{\lambda}}{\partial w}(P_*(y), w_*(y), y)(w_2(y) - w_1(y))\right|\right] y dy
\nonumber\\
 &\leq& \int_0^x 
\left[ C_1(R) |P_2(y) - P_1(y)| + C_2(R) |w_2(y) - w_1(y)| \right] y dy,
\label{Eq:TPDiffEst}
\end{eqnarray}
with the constants
\begin{equation*}
C_1(R) := \max\limits_{\substack{\frac{1}{2} \leq P \leq 1 \\ w_0 \leq w \leq 1 \\ 0 \leq y \leq R }}
\left|\frac{\partial F_{\lambda}}{\partial P}(P, w, y)\right|,
\qquad
C_2(R) := \max\limits_{\substack{\frac{1}{2} \leq P \leq 1 \\ w_0 \leq w \leq 1 \\ 0 \leq y \leq R }}
\left|\frac{\partial F_{\lambda}}{\partial w}(P, w, y)\right|.
\end{equation*}
Taking the supremum over $x$ on both sides of the inequality~(\ref{Eq:TPDiffEst}) one obtains the estimate
\begin{equation}
\|TP_2 - TP_1\|_{\infty} \leq \frac{R^2}{2}
\left[ C_1(R)\|P_2 - P_1\|_\infty + C_2(R) \|w_2 - w_1\|_\infty\right],
\label{Eq:TPDiffEstBis}
\end{equation}
for all $P_1,P_2\in A_R$ and $w_1,w_2\in B_R$. Furthermore, using the definition~(\ref{Eq:WA}), one obtains in a similar manner the estimate
\begin{equation}
|w_2(x) - w_1(x)| \leq \frac{3}{x^3}\int_0^x |e(P_2(y)) - e(P_1(y))| y^2\, dy
\leq C_3 \| P_2 - P_1 \|_\infty,
\label{Eq:wDiffEst}
\end{equation}
with the constant
\begin{equation*}
C_3 := \max\limits_{\frac{1}{2} \leq P \leq 1}\left| \frac{de}{dP}(P)\right|,
\end{equation*}
where we have used that fact that $e: [1/2, 1] \to \Real$ is a continuously differentiable function due to the properties of the function $\varepsilon(P)$ defined in~(\ref{Eq:epsilonp}). Combining the two estimates~(\ref{Eq:TPDiffEstBis},\ref{Eq:wDiffEst}) one obtains, finally
\begin{equation}
\|TP_2 - TP_1\|_\infty \leq L(R)\|P_2 - P_1\|_\infty,\qquad
L(R) := \frac{R^2}{2}\left[ C_1(R) + C_2(R) C_3 \right],
\end{equation}
for all $P_1,P_2\in A_R$. Since $C_1(R)$ and $C_2(R)$ decrease with $R$, it is clear that one can choose $R > 0$ small enough such that $L(R) < 1$ and $T: A_R\to A_R$ describes a contraction on $A_R$. Now we can use the contraction mapping principle  (Theorem~\ref{Thm:Banach}) to show:

\begin{Teo}
\label{Thm:LocalExistence}
For small enough $R > 0$, there exists a unique, continuously differentiable solution $P: (0,R)\to \Real$ of the dimensionless TOV equation~(\ref{Eq:TolmanABis}) satisfying $\lim\limits_{x\to 0} P(x) = 1$.
\end{Teo}

\begin{proof}
Theorem~\ref{Thm:Banach} and the previous observations guarantee that for small enough $R > 0$ the map $T$ has a unique fixed point $P$ in $A_R$. Since $TP: (0,R)\to \Real$ is differentiable, $P = TP$ is differentiable as well and differentiating both sides of the equation $P(x) = TP(x)$ with respect to $x$ one finds that Eq.~(\ref{Eq:TolmanABis}) is satisfied for all $0 < x < R$, and hence $dP/dx$ is also continuous.

Regarding the uniqueness property, if $\tilde{P}: (0,R)\to \Real$ was another continuously differentiable solution of Eq.~(\ref{Eq:TolmanABis}) such that $\lim\limits_{x\to 0}\tilde{P}(x) = 1$, then $\tilde{P}$ would also be a fixed point of $T$ and hence would agree with $P$. 
\end{proof}

Finally, $\phi$ is obtained by integrating both sides of Eq.~(\ref{Eq:TolmanA}):
\begin{equation}
\phi(x) = \phi_c + \int_0^x \frac{w(y) + 3\lambda P(y)}{1 - 2\lambda y^2 w(y)} y dy,\qquad
0\leq x < R,
\end{equation}
with a constant of integration $\phi_c$ denoting the central value of $\phi$. If the solution exists globally, one can adjust this constant such that $\phi(x)\to 0$ for $x\to \infty$. Equivalently, if a global solution with finite radius $x_* > 0$ exists (sufficient conditions for this to occur will be discussed in the next section), one can choose the value of $\phi_c$ such that $\phi(x_*)$ matches its Schwarzschild value $\phi(x_*) = \frac{1}{2} \log\left(1 - \frac{2M}{\ell x_*}\right)$, with $M := \ell\lambda x_*^3 w(x_*)$ the total mass of the configuration.

In this way, one obtains a unique, continuously differentiable solution $(\phi(x),P(x))$ of Eqs.~(\ref{Eq:TolmanA}) on a small interval $(0,R)$ near the center with the required boundary conditions $\phi(0) = \phi_c$ and $P(0) = 1$. Moreover, with some algebra work one can show that the original Euler-Einstein equations~(\ref{Eq:Einstein1},\ref{Eq:Einstein2},\ref{Eq:Einstein3}) are satisfied.

\section{Global existence of finite radius solutions and Buchdahl bound}
\label{Sec:GlobalExistence}

In the previous section we proved the existence of a unique solution $P: (0,R)\to \Real$ of the dimensionless TOV equation~(\ref{Eq:TolmanABis}) on a small interval $(0,R)$, which satisfies the required boundary condition $\lim\limits_{x\to 0} P(x) = 1$ at the center, see Theorem~\ref{Thm:LocalExistence}. In this section, we show that under suitable hypotheses on the equation of state, this solution can be extended to an interval $(0,x_*)$ with $x_* > R$ describing the surface of the star, which is characterized by the condition $\lim\limits_{x\to x_*} P(x) = 0$ of vanishing pressure.

To prove this result, we define
\begin{align*}
x_*  & := \sup\bigg\{ x_1 > 0 \; \biggr\rvert \; P: (0, x_1) \to \Real 
\hbox{ is a continuously differentiable solution of Eq.~(\ref{Eq:TolmanABis}) satisfying }
\lim\limits_{x\to 0} P(x) = 1 \\
& \qquad\qquad\qquad\qquad
\hbox{and such that $0 < P(x) \leq 1$ and $1 - 2\lambda x^2 w(x) > 0$ for all $x\in (0,x_1)$} 
\bigg\}.
\end{align*}
According to Theorem~\ref{Thm:LocalExistence}, $x_* > 0$ is well-defined. There are two alternatives. Either
\begin{itemize}
\item[(a)] $x_* < \infty$ is finite, or
\item[(b)] $x_* = \infty$ is infinite.
\end{itemize}
Moreover, since $dP/dx < 0$, $P(x)$ is a monotonously decreasing function and case (a) occurs either if
\begin{itemize}
\item[(a.1)] $\displaystyle\lim_{x\to x_*}{[1 - 2\lambda x^2 w(x)]} > 0$ and $\displaystyle\lim_{x \to x_*}{P(x)} = 0$, or if
\item[(a.2)] $\displaystyle\lim_{x\to x_*}{[1 - 2\lambda x^2 w(x)]} = 0$.
\end{itemize}
The central result of this section is to show that under the conditions $(i)$--$(iv)$ in section~\ref{sec:EquationState}, only the case (a.1) can occur if $\gamma_1 > 4/3$, which means that the local solution has a unique extension describing a star of finite radius $R_* = \ell x_* > 0$. The strategy of the proof is the following: first, we eliminate case (b), i.e. we exclude the possibility of a star with infinite extension. Subsequently, we eliminate case (a.2) by proving that the averaged density function $w(x)$ cannot grow too fast to make the denominator in Eq.~(\ref{Eq:TolmanABis}) zero. As a by-product of this result, we will also obtain a bound on the compactness ratio
\begin{equation}
\frac{2m(r)}{r} = 2\lambda x^2 w(x),
\end{equation}
which shows that it is, in fact, not only smaller than one (as required to eliminate case (a.2)) but even smaller than $8/9$ for all $0 < x < x_*$. In particular, this implies that the compactness ratio at the surface of the star $r \to R_*$ is bounded from above by the well-known Buchdahl value $8/9$.

We start with the following theorem which eliminates case (b):

\begin{Teo}
\label{Thm:FiniteRadius}
Suppose the conditions $(i)$--$(iv)$ in section~\ref{sec:EquationState} are satisfied with the lower adiabatic bound $\gamma_1 > 4/3$. Then $x_* < \infty$ is finite.
\end{Teo}

\begin{proof}
We suppose that $x_* = \infty$ is infinite and show that this leads to a contradiction. Since $x_* = \infty$ implies that $P$ is bounded, and since $P$ is monotonously decreasing, the limit
\begin{equation}
\label{Eq:Pinf}
P_\infty := \lim\limits_{x\to \infty} P(x) \geq 0
\end{equation}
exists. The remainder of the proof is based on the following two simple lemmas whose proofs will be given further below. The first lemma shows that $P_\infty$ must be zero:

\begin{Lem}
\label{Lem:1}
Suppose $x_* = \infty$. Then $P_\infty = 0$.
\end{Lem}

The second lemma provides a lower bound on the energy density which will be key in the proof of the theorem:

\begin{Lem}
\label{Lem:2}
Any equation of state fulfilling the conditions $(i)$--$(iv)$ in section~\ref{sec:EquationState} satisfies the following estimate: there are constant $C > 0$ and $P_1 > 0$ such that
\begin{equation}
\label{Eq:Estimatione}
e(P) \geq CP^{1/\gamma_1},
\end{equation}
for all $0 \leq P\leq P_1$.
\end{Lem}

We now return to the proof of Theorem~\ref{Thm:FiniteRadius} and show that $x_* = \infty$ and $P_\infty = 0$ leads to a contradiction if $\gamma_1 > 4/3$. To this purpose, we use Eq.~(\ref{Eq:TolmanABis}) to estimate
\begin{equation}
-\frac{1}{e(P(x)) + \lambda P(x)} \frac{d}{dx} P(x) \geq x w(x)
\end{equation}
for all $x > 0$. Integrating both sides of this inequality yields
\begin{equation}
\label{Eq:Est}
- \int_x^\infty \frac{1}{e(P(y)) + \lambda P(y)} \frac{dP}{dy}(y) dy \geq \int_x^\infty w(y) y dy.
\end{equation}
Using the variable substitution $P = P(y)$ and the estimate~(\ref{Eq:Estimatione}), the integral on the left-hand side can be rewritten and estimated according to
\begin{equation}
- \int_x^\infty \frac{1}{e(P(y)) + \lambda P(y)} \frac{dP}{dy}(y) dy
 = \int_0^{P(x)} \frac{dP}{e(P) + \lambda P} 
 \leq \int_0^{P(x)} \frac{dP}{C P^{1/\gamma_1}} 
  = \frac{P(x)^{1 - 1/\gamma_1}}{C(1 - 1/\gamma_1)},
\end{equation}
for all large enough $x\geq x_1$, such that $P(x_1)\leq P_1$. This yields the following lower bound on $P$:
\begin{equation}
\label{Eq:EstimationLHS}
P(x)^{1 - 1/\gamma_1} 
 \geq - C_1\int_x^\infty \frac{1}{e(P(y)) + \lambda P(y)} \frac{dP}{dy}(y) dy,
\end{equation}
with $C_1 := C(1 - 1/\gamma_1) > 0$ a constant. Next, we estimate the integral on the right-hand side of Eq.~(\ref{Eq:Est}). Recalling that $\overline{m}(x) := x^3 w(x)$ is proportional to the mass function, which is an increasing function of $x$, we obtain
\begin{equation}
\label{Eq:EstimationRHS}
\int_x^\infty w(y) y dy = \int_x^\infty \frac{\overline{m}(y)}{y^2} dy
 \geq \overline{m}(x) \int_x^\infty\frac{dy}{y^2} = \frac{\overline{m}(x)}{x} = x^2 w(x),
\end{equation}
for all $x > 0$. The three estimates~(\ref{Eq:Est},\ref{Eq:EstimationLHS},\ref{Eq:EstimationRHS}) imply the following inequality between $P$ and $w$:
\begin{equation}
P(x)^{1 - 1/\gamma_1} \geq C_1 x^2 w(x)
\label{Eq:Est2}
\end{equation}
for all $x\geq x_1$. Combining this with the estimate $w(x) \geq e(P(x))$ (which follows directly from the definition~(\ref{Eq:WA}) of $w(x)$ and the monotonicity properties of $e$ and $P$) and the key estimate~(\ref{Eq:Estimatione}) yields
\begin{equation}
P(x)^{1 - 2/\gamma_1} \geq C_2 x^2
\label{Eq:Est3}
\end{equation}
for all $x\geq x_1$, with the new constant $C_2 := C C_1 = C^2(1 - 1/\gamma_1) > 0$. This already yields a contradiction for $\gamma_1 \geq 2$, since in this case the left-hand side converges to zero (or stays constant if $\gamma_1 = 2$) while the right-hand side goes to infinity as $x\to \infty$. This proves the theorem for $\gamma_1\geq 2$.

It remains to analyze the case $4/3 < \gamma_1 < 2$. For this, we use again the key estimate~(\ref{Eq:Estimatione}) and the fact that $e(P(x)) \leq w(x)$, obtaining $P(x)^{1/\gamma_1} \leq C^{-1} w(x)$, or
\begin{equation}
\left[\frac{w(x)}{C}\right]^{\gamma_1} \geq P(x)
\end{equation}
for all $x\geq x_1$. Combining this with the inequality~(\ref{Eq:Est2}) yields
\begin{equation}
w(x)^{\gamma_1 -2} \geq C_3 x^2
\end{equation}
for all $x\geq x_1$ with the positive constant $C_3 = C_1 C^{\gamma_1-1}$. Since $w(x) = \overline{m}(x)/x^3$ and $\gamma_1 - 2 < 0$ this can be rewritten as
\begin{equation}
\overline{m}(x)^{2 - \gamma_1} \leq \frac{1}{C_3 x^{3\gamma_1 - 4}},
\end{equation}
for $x\geq x_1$. However, since $4/3 < \gamma_1  < 2$ this leads to a contradiction since in the limit $x\to \infty$ the right-hand side converges to $0$ while the mass function $\overline{m}(x)$ is positive and increasing. This concludes the proof of the theorem.
\end{proof}

\begin{proof}[Proof of Lemma~\ref{Lem:1}]
Again, the proof is by contradiction. If $P_\infty \neq 0$, then the function $P$ would satisfy
$P(x) \geq P_{\infty} > 0$ for all $x > 0$, and since $e(P)$ is monotonously increasing, this would imply that $e(P(x)) \geq e(P_\infty) =: e_\infty > 0$ for all $x > 0$. According to Eq.~(\ref{Eq:WA}) this would yield $w(x) \geq e_{\infty} > 0$ for all $x > 0$, which in turn would imply that
\begin{equation}
1 - 2\lambda x^2 w(x) \leq 1 - 2\lambda x^2 e_\infty
\end{equation}
for all $x > 0$. However, this would contradict the assumption $x_* = \infty$ which requires $1 - 2\lambda x^2 w(x) > 0$ for all $x > 0$. Therefore, we must have $P_\infty = 0$ as claimed.
\end{proof}

\begin{proof}[Proof of Lemma~\ref{Lem:2}]
For the proof of this lemma, we use the inequality~(\ref{Eq:pInequality}) from section~\ref{sec:EquationState}, which implies
%
%
\begin{equation}
\label{Eq:n}
n \geq n_2\left[ \frac{p(n)}{p(n_2)}\right]^{1/\gamma_1}
\end{equation}
for all small enough $n_2\geq n > 0$. Using the assumptions $(i)$ and $(iv)$ from section~\ref{sec:EquationState} and the estimate~(\ref{Eq:n}) in the expression~(\ref{Eq:epsilonp}) for $\varepsilon(p)$ one obtains,
\begin{equation}
\varepsilon(p) \geq ne_0 \geq  n_2 e_0  \left[ \frac{p(n)}{p(n_2)}\right]^{1/\gamma_1},
\end{equation}
for all small enough $0 < n\leq n_2$. Setting $C_2 := n_2 e_0/p_2^{1/\gamma_1}$ with $p_2 := p(n_2)$ it follows from this that
\begin{equation}
\varepsilon(p) \geq C_2 p^{1/\gamma_1}
\end{equation}
for all $0 < p\leq p_2$. Since $\varepsilon(p) = \varepsilon_c e(P)$ and $p = p_c P$ the lemma follows. 
\end{proof}

To conclude the global existence proof, it remains to eliminate case (a.2). In fact, we obtain a stronger result which shows that for all $0 < x < x_*$, one must have $1 - 2\lambda x^2 w(x) = 1 - 2m(r)/r < 1/9$:

\begin{Teo}
\label{Thm:Buchdahl}
Let $P: (0,x_*)\to \Real$ be the maximally extended continuously differentiable solution of the dimensionless TOV Eq.~(\ref{Eq:TolmanABis}) such that $\lim\limits_{x\to 0} P(x) = 1$, $0 < P(x) < 1$ and $1 - 2\lambda x^2 w(x) > 0$ for all $0 < x < x_*$. Then, $2m(r)/r = 2\lambda x^2 w(x) < 8/9$ for all $0 < x < x_*$.
\end{Teo}

\begin{proof}
The proof is a straightforward generalization to arbitrary radius $r\in (0,R_*)$ of standard arguments used to establish the Buchdahl bound, see for instance section 6.2 in Ref.~\cite{Wald}. For this, we set $r = \ell x$, $m(r) := \ell\lambda x^3 w(x)$, $\Psi(r) := -\frac{1}{2}\log[1 - 2m(r)/r) ]$ and use the fact that the Einstein equations~(\ref{Eq:Einstein1},\ref{Eq:Einstein2},\ref{Eq:Einstein3}) are satisfied. Subtracting Eq.~(\ref{Eq:Einstein2}) from Eq.~(\ref{Eq:Einstein3}) yields
\begin{equation}
\left[\Phi'' + \Phi'(\Phi' - \Psi') - \frac{\Phi' + \Psi'}{r}\right]e^{-2\Psi} - \frac{1}{r^2}\left(e^{-2\Psi} - 1\right) 
 = 0.
\end{equation}
Dividing both sides by $r$ one can rewrite this as the following identity:
\begin{equation}
e^{-\Phi(r) - \Psi(r)}\left[\frac{\Phi'(r)}{r}e^{\Phi(r) - \Psi(r)}\right]' = \left[\frac{m(r)}{r^3}\right]'.
\label{Eq:Identity}
\end{equation}
Since $m(r)/r^3$ is proportional to the mean density, which is by itself proportional to $w(x)$, and since $x\frac{dw}{dx} = 3[ e(P(x)) - w(x) ]\leq 0$, the mean density is a non-increasing function. Therefore, it follows from Eq.~(\ref{Eq:Identity}) that
\begin{equation}
\label{Eq:Des1}
\left[\frac{\Phi '(r)}{r}e^{\Phi(r) - \Psi(r)}\right]' \leq 0.
\end{equation}
Next, let $0 < r < r_2 < R_* = \ell x_*$. Then, it follows that
\begin{equation}
\frac{\Phi'(r)}{r} e^{\Phi(r) - \Psi(r)} \geq \frac{\Phi'(r_2)}{r_2} e^{\Phi(r_2) - \Psi(r_2)} 
 = \frac{m(r_2) + 4\pi r_2^3 p(r_2)}{r_2^3\left[ 1 - \frac{2m(r_2)}{r_2} \right]} 
  e^{\Phi(r_2) - \Psi(r_2)},
\end{equation}
where we have used Eq.~(\ref{Eq:Phi}) to eliminate $\Phi'(r_2)$. Since $p(r_2)\geq 0$ and
\begin{equation}
1 - \frac{2m(r_2)}{r_2} = e^{-2\Psi(r_2)},
\label{Eq:a2}
\end{equation}
this inequality leads to
\begin{equation}
\Phi'(r) e^{\Phi(r)} \geq r e^{\Psi(r)}\frac{m(r_2)}{r_2^3} e^{\Phi(r_2) + \Psi(r_2)}.
\end{equation}
Integrating both sides from $r = 0$ to $r_2$ yields
\begin{equation}
e^{\Phi(r_2)} - e^{\Phi(0)} \geq e^{\Phi(r_2) + \Psi(r_2)}
\frac{m(r_2)}{r_2^3} \int_0^{r_2} \frac{r dr}{\sqrt{1 - \frac{2m(r)}{r}}}.
\end{equation}
To estimate the integral on the right-hand side, we use again the fact that $m(r)/r^3$ is a non-increasing function, such that $2m(r) \geq 2m(r_2) r^3/r_2^3$ for all $0\leq r\leq r_2$, and obtain
\begin{equation}
e^{\Phi(r_2)} - e^{\Phi(0)} \geq e^{\Phi(r_2) + \Psi(r_2)}
 \frac{m(r_2)}{r_2^3}\int_0^{r_2} \frac{r dr}{\sqrt{1 - \frac{2m(r_2)}{r_2^3} r^2}}
 = \frac{1}{2} e^{\Phi(r_2)} \left[ e^{\Psi(r_2)} - 1 \right],
\label{Eq:Des}
\end{equation}
where we have used Eq.~(\ref{Eq:a2}) again. Eq.~(\ref{Eq:Des}) implies that
\begin{equation}
0 < 2e^{\Phi(0)} \leq e^{\Phi(r_2)} \left[ 3 - e^{\Psi(r_2)} \right],
\end{equation}
which immediately yields the desired result:
\begin{equation}
1 - \frac{2m(r_2)}{r_2} = e^{-2\Psi(r_2)} > \frac{1}{9}.
\end{equation}
\end{proof}

\section{A numerical example}
\label{Sec:Numerical}

In the previous sections we have shown that for a given equation of state fulfilling the conditions $(i)$--$(iv)$ in section~\ref{sec:EquationState} with the lower adiabatic bound $\gamma_1 > 4/3$, there exists for each value of $p_c/\varepsilon_c > 0$ a unique solution of the TOV equation which describes a relativistic, spherical and static star of finite radius $R$ and mass $M$. In this section, we show by means of numerical calculation how to obtain the quantitative properties of the star, including the values of $R$ and $M$, the compactness ratio $2M/R$ and the pressure profile. For the sake of illustration we focus on the specific case of a polytropic equation of state of the form
\begin{equation}
p(n) = K n^{\gamma}
\label{Eq:polytrope}
\end{equation}
with $K$ a positive constant and $\gamma$ the adiabatic index which, in the results shown below, is fixed to the value $5/3$. As explained in appendix~\ref{App:StatFis}, this value corresponds to the low temperature and density limit of a monoatomic ideal gas. Integrating the first law for an isentropic fluid yields the corresponding expression for the energy density
\begin{equation}
\varepsilon(p) = n e_0 + \frac{K}{\gamma-1} n^\gamma 
 = e_0\left( \frac{p}{K} \right)^{1/\gamma} + \frac{p}{\gamma-1}.
\label{Eq:polytrope_eps}
\end{equation}
Rewritten in terms of the dimensionless quantities defined in Eq.~(\ref{Eq:Dimensionless}) and using the fact that $e(1) = 1$, this yields
\begin{equation}
e(P) = \frac{(\gamma-1-\lambda) P^{1/\gamma} + \lambda P}{\gamma-1},\qquad
0 < \lambda = \frac{p_c}{\varepsilon_c} < \gamma-1.
\end{equation}
(Note that for the case of a monoatomic gas one should also have $p_c/\varepsilon_c \ll 1$ in the low temperature limit, so that the example studied in this section is most probably not physically realistic for values of $\lambda$ lying close to $\gamma-1$.)

To perform the numerical integration of the TOV equation, we convert the integral equation~(\ref{Eq:WA})  for the dimensionless mean density field $w$ into the differential equation
\begin{equation}
\frac{d}{dx} w(x) = -\frac{3}{x}\left[ w(x) - e(P(x)) \right],
\label{Eq:dw}
\end{equation}
which is numerically integrated along with the dimensionless TOV equation~(\ref{Eq:TolmanABis}) using a standard fourth-order accurate Runge-Kutta scheme (see, for instance, section 7.5 in Ref.~\cite{oSmT12} and references therein). The integration is started at the center $x = 0$, where the right-hand side of Eq.~(\ref{Eq:dw}) is replaced with $0$, owing to the fact that both functions $w(x)$ and $P(x)$ behave as $1 + {\cal O}(x^2)$ near $x = 0$. (This can be inferred from the local existence theorem in section~\ref{Sec:LocalExistence}, the fixed point formula~(\ref{Eq:IntTOV}) and the definition of $w$ in Eq.~(\ref{Eq:WA}).) The integration is stopped as soon as $P$ becomes negative, which yields the dimensionless radius $R/\ell = x_*$ and the dimensionless total mass $M/\ell = 
\lambda x_*^3 w(x_*)$ of the star, up to a numerical error. (This error is monitored by varying the stepsize $\Delta x$ of the integrator.) Using Eqs.~(\ref{Eq:lDef}) and (\ref{Eq:polytrope_eps}) one finds that the length scale $\ell$ is given by
\begin{equation}
\ell = \ell_0 \lambda^{-\frac{2-\gamma}{2(\gamma-1)}}
\left(  1 - \frac{\lambda}{\gamma-1} \right)^{\frac{\gamma}{2(\gamma-1)}},\qquad
\ell_0 := \sqrt{\frac{3}{4\pi}}\left( \frac{K}{e_0^\gamma} \right)^{\frac{1}{2(\gamma-1)}},
\label{Eq:ell0}
\end{equation}
and hence we shall specify the results in terms of the alternative length scale $\ell_0$ which is independent of $\lambda$.

The results of the numerical integration for different values of $\lambda$ in the admissible range $0 < \lambda < \gamma - 1$ are shown in Table~\ref{Tab:Polytrope} and in Figs.~\ref{Fig:MassCompactness},\ref{Fig:MvsR} and~\ref{Fig:Profiles}. Note that for small values of $\lambda$ the mass increases while the radius of the star decreases as $\lambda$ grows, giving rise to more compact stars. However, as $\lambda$ continues to grow this trend is halted and $M/\ell_0$ reaches a maximum at about $\lambda\approx 0.12$ after which it starts decaying as $\lambda$ continues to grow until it reaches a local minimum around $\lambda\approx 0.5$ and starts growing again until reaching another local maximum. Similarly, the radius $R/\ell_0$ decreases until it reaches a local minimum at about $\lambda\approx 0.4$ after which it increases until reaching a local maximum. This behavior gives rise to the spiral structure shown in Fig.~\ref{Fig:MvsR}.

In the Newtonian limit $\lambda\to 0$, one may compare our results with the corresponding results from the Lane-Emden equation (see for instance section 3.3 in~\cite{Shapiro})
\begin{equation}
\frac{R}{\ell} = \frac{a}{\ell}\xi_1,\qquad
\frac{M}{\ell} = 3\frac{a^3}{\ell^3}\lambda\xi_1^2|\Theta'(\xi_1)|,\qquad
\frac{a^2}{\ell^2} = \frac{1}{3}\frac{\gamma}{\gamma-1}.
\end{equation}
For the present example $\gamma = 5/3$ one finds $\xi_1 \approx 3.65$, $\xi_1^2|\Theta'(\xi_1)| \approx 2.71$ and $a/\ell = \sqrt{5/6}$, which yields
\begin{equation}
\frac{R}{\ell} \approx 3.33,\qquad
\frac{M}{\ell} \approx 6.18\lambda,
\end{equation}
and compares well with the corresponding values in Table~\ref{Tab:Polytrope} for small $\lambda$.

Finally, we note again from the plots in Fig.~\ref{Fig:Profiles} that the relativistic stars with high $\lambda$ are much more compact than their Newtonian counterparts. We also note that although the compactness ratio $2M/R$ at the surface reaches a maximum at about $\lambda\approx 0.3$, the maximum of the local compactness ratio $2m(r)/r$ occurs inside (and not at the surface of) the star, and this maximum seems to be growing monotonously with $\lambda$. In all cases this maximum is less than $8/9$, as predicted by the local Buchdahl bound proven in Theorem~\ref{Thm:Buchdahl}. (Note that the Newtonian equations predict a compactness ratio of $2M/R\approx 3.71\lambda$ which can be larger than one.)

\begin{table*}[h]
\caption{Results for the dimensionless radius $R/\ell = x_*$, dimensionless total mass $M/\ell = \lambda x_*^3 w(x_*)$ and compactness ratio $2M/R = 2\lambda x_*^2 w(x_*)$ at the surface of the star for the polytropic equation of state~(\ref{Eq:polytrope}) and different values of $\lambda$. Also shown are the radii $R/\ell_0$ and masses $M/\ell_0$ in terms of the physical scale $\ell_0$ defined in Eq.~(\ref{Eq:ell0}) which is independent of $\lambda$. The stepsize used to produce these results is $\Delta x = 0.005$, and three significant figures are shown.
}
\begin{tabular}{c||c|c|c|c||c}
$\lambda$  & $R/\ell$ & $M/\ell$ & $R/\ell_0$ & $M/\ell_0$ & $2M/R$ \\
\hline
 $0.001$ & $3.33$ & $0.0615$ & $18.7$ & $0.0345$ & $0.00370$ \\
 $0.01$  & $3.29$  & $0.0582$ & $10.2$ & $0.181$ & $0.0353$ \\
 $0.05$  & $3.16$  & $0.232$ & $6.06$ & $0.446$ & $0.147$ \\
 $0.1$  & $3.08$  & $0.368$ & $4.47$ & $0.535$ & $0.239$ \\
 $0.2$  & $3.16$  & $0.524$ & $3.03$ & $0.502$ & $0.331$ \\
 $0.3$  & $3.68$  & $0.640$ & $2.36$ & $0.410$ & $0.348$ \\
 $0.4$  & $5.24$  & $0.812$ & $2.10$ & $0.325$ & $0.310$ \\
 $0.5$  & $11.3$  & $1.34$ & $2.37$ & $0.282$ & $0.238$ \\
 $0.6$  & $42.8$  & $5.12$ & $2.74$ & $0.327$ & $0.239$ \\
 $0.65$  & $234$  & $28.9$ & $2.59$ & $0.320$ & $0.246$ \\
\end{tabular}
\label{Tab:Polytrope}
\end{table*}

\begin{figure}[H]
\centerline{\includegraphics[width=9cm]{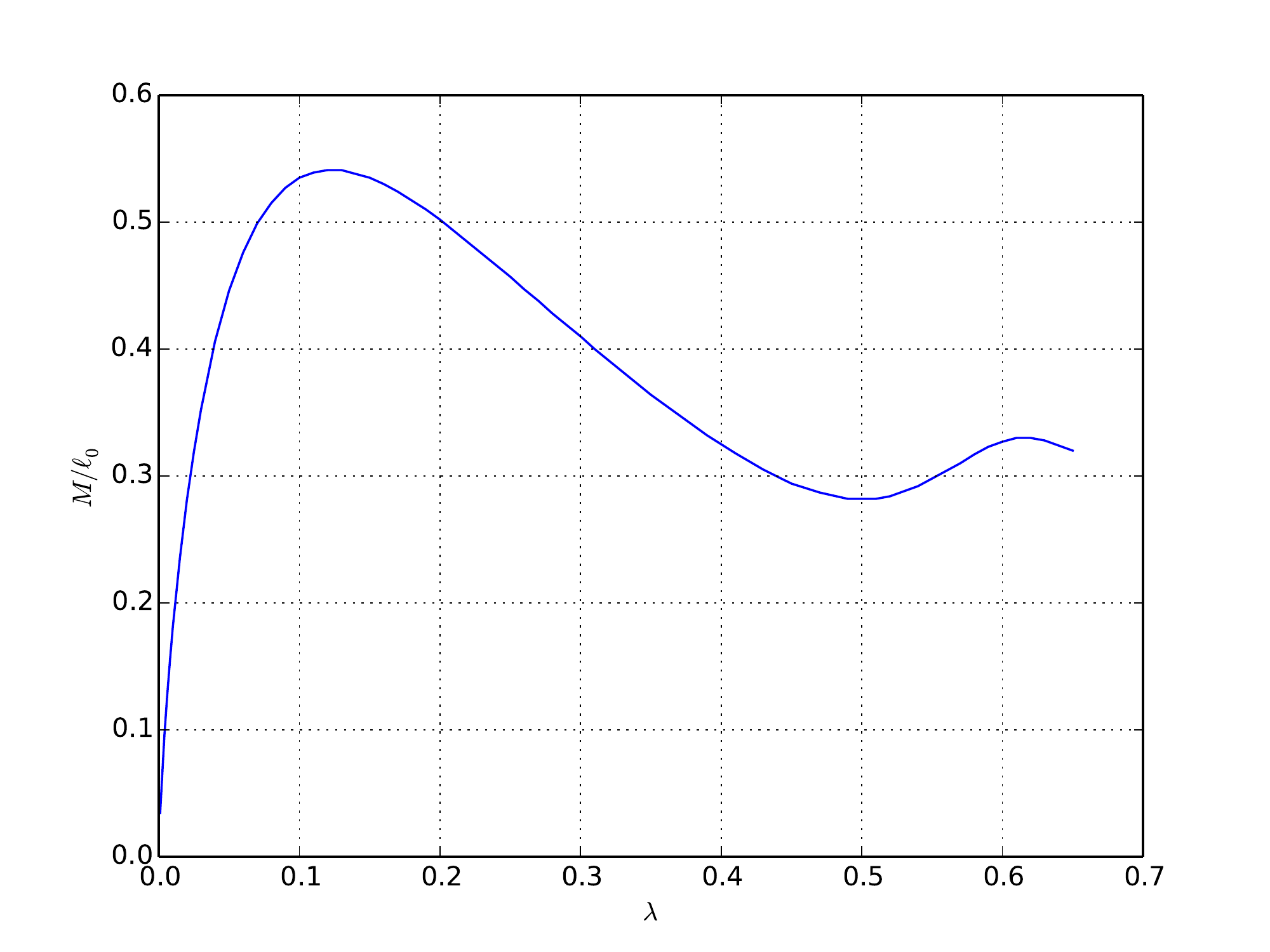}
\includegraphics[width=9cm]{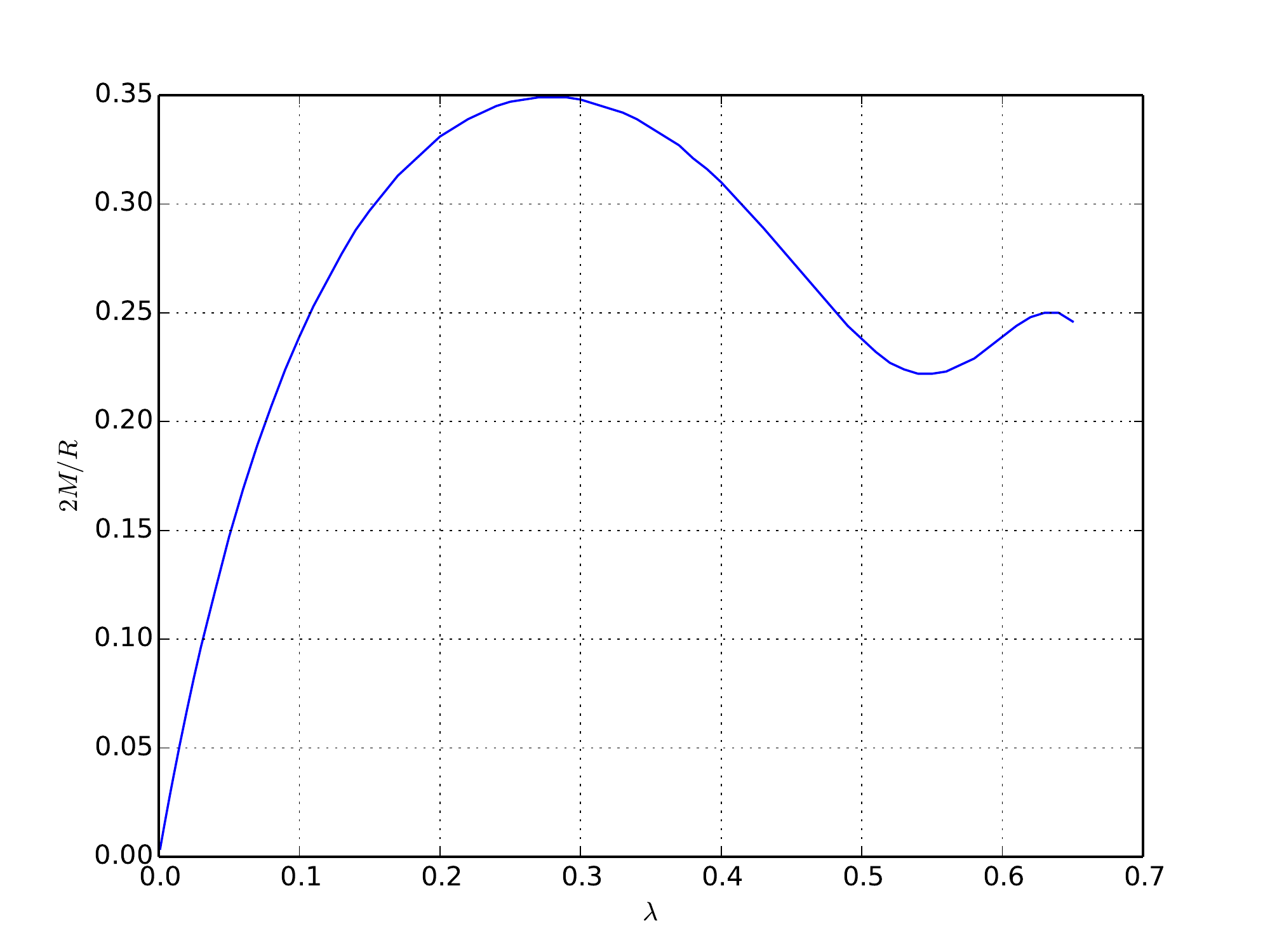}
}
\caption{Plots of the total mass $M/\ell_0$ (left panel) and the compactness ratio $2M/R$ at the surface of the star  (right panel) as a function of $\lambda$.}
\label{Fig:MassCompactness}
\end{figure}

\begin{figure}[H]
\centerline{\includegraphics[width=11cm]{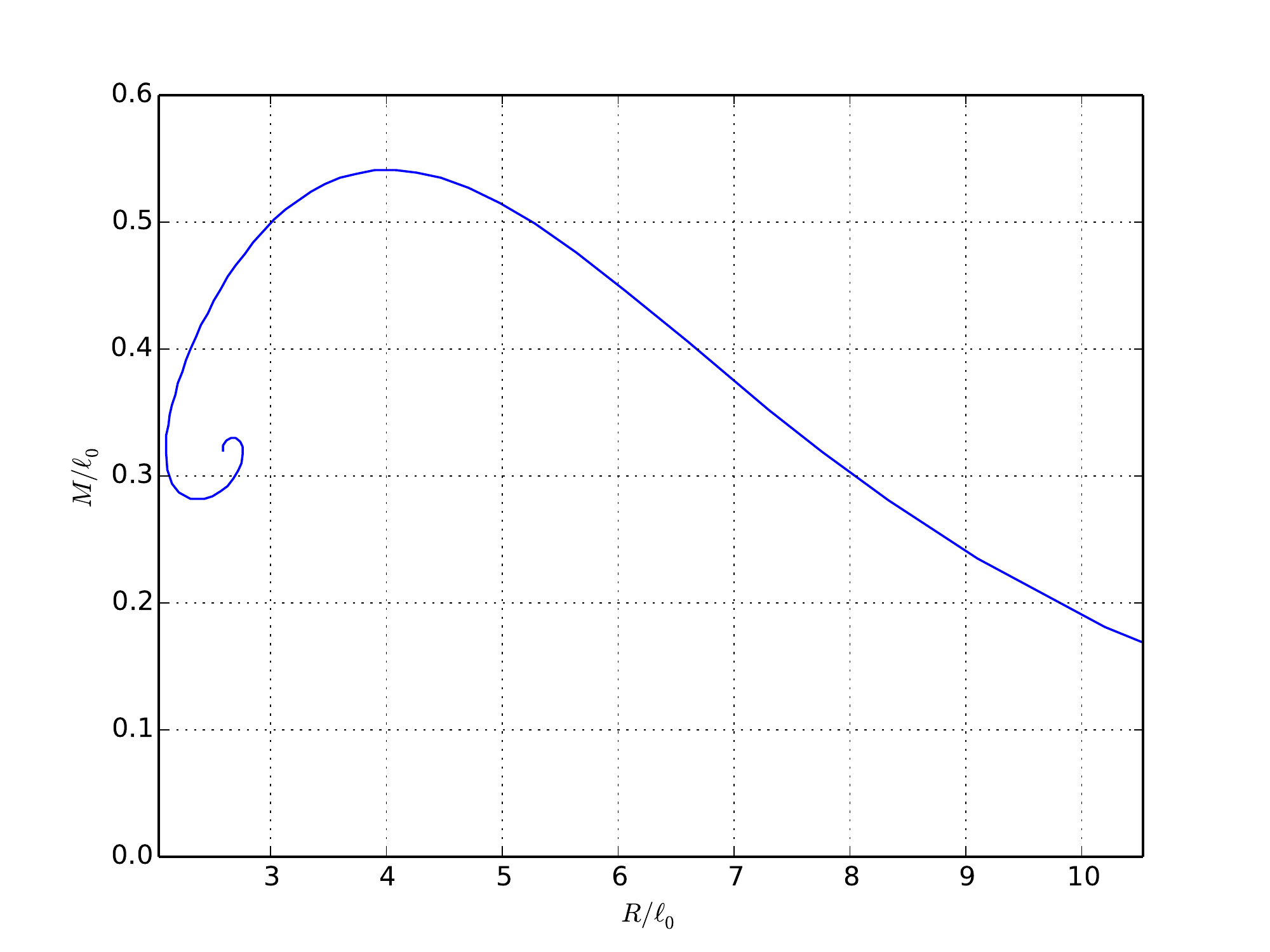}}
\caption{The total mass $M/\ell_0$ vs. radius $R/\ell_0$ for different values of $\lambda$.}
\label{Fig:MvsR}
\end{figure}

\begin{figure}[H]
\centerline{\includegraphics[width=9cm]{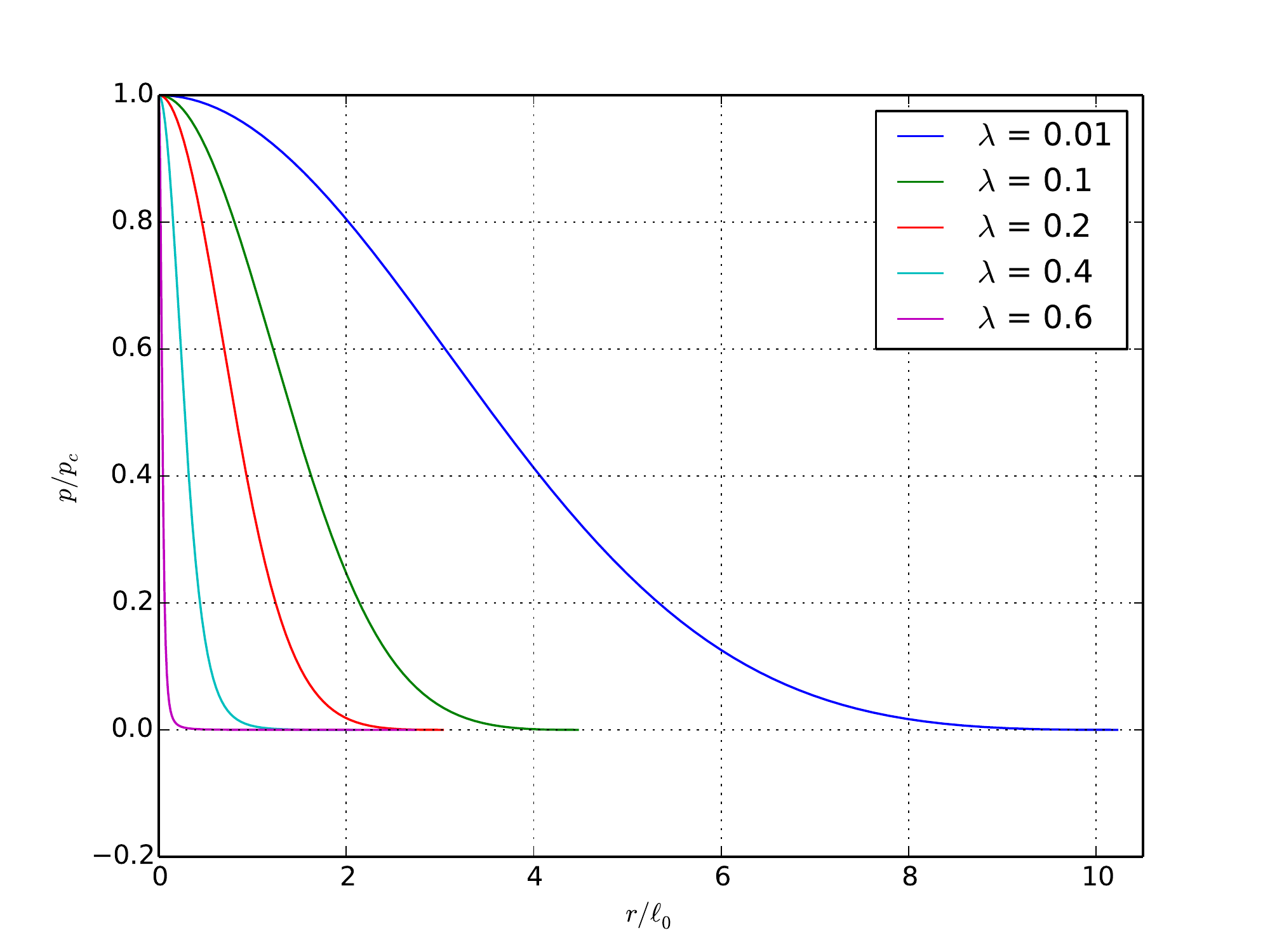}
\includegraphics[width=9cm]{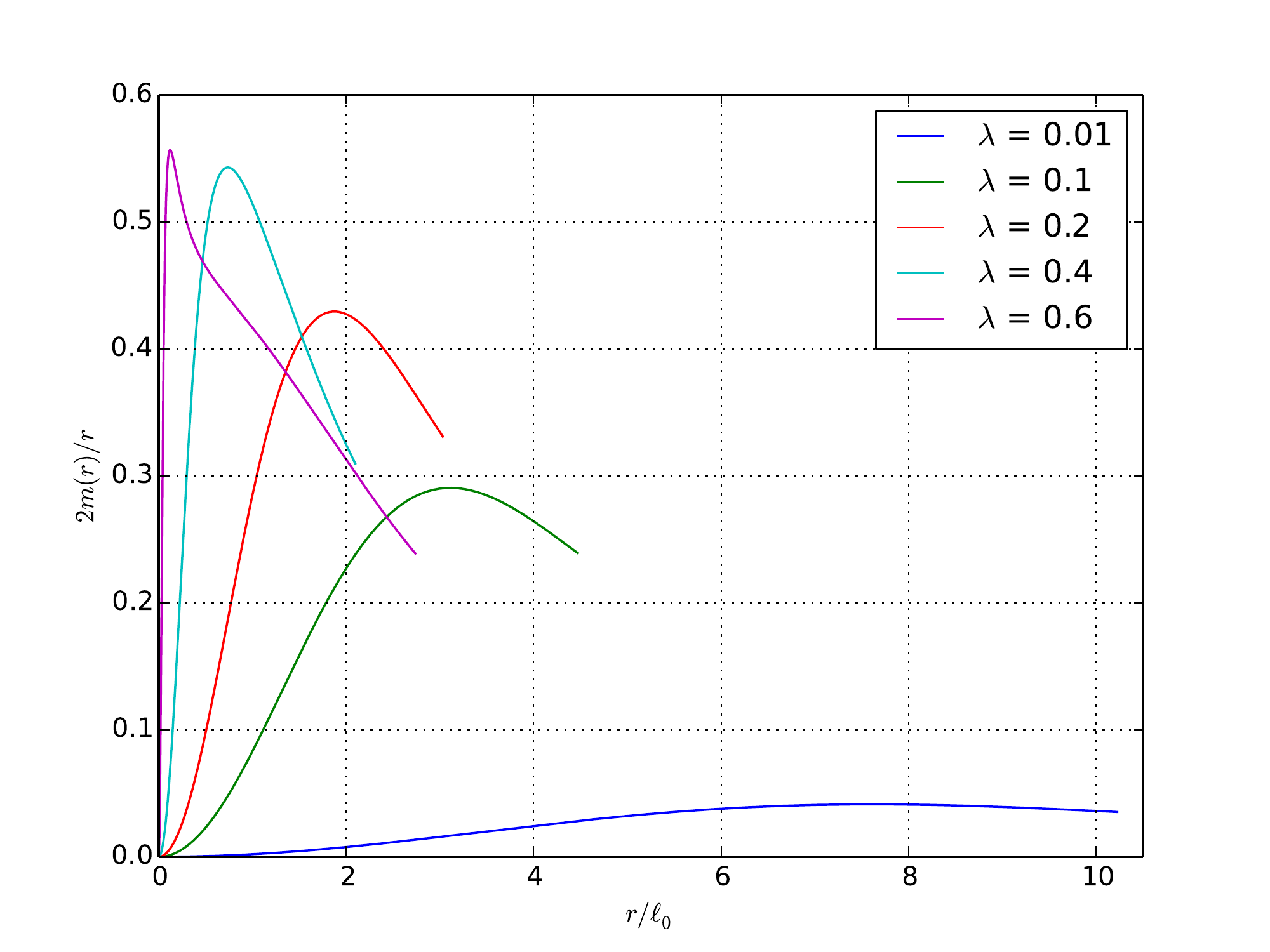}
}
\caption{Plots of the dimensionless pressure $p/p_c = P$ (left panel) and the local compactness ratio $2m(r)/r = 2\lambda x^2\omega(x)$ (right panel) as a function of the dimensionless radius $r/\ell_0 = x\ell/\ell_0$ for different values of $\lambda$. As is visible from these plots the stars become more compact as $\lambda$ increases, with the maximum of the local compactness ratio lying inside the star.}
\label{Fig:Profiles}
\end{figure}

\section{Summary and conclusions}
\label{Sec:Conclusions}

In this article, we have given a systematic derivation of the TOV equation, starting from the most general static and spherically symmetric ansatz for the metric and fluid fields which allows one to reduce the Euler-Einstein system to a set of ordinary differential equations. Under the assumptions on the equation of state discussed in section~\ref{sec:EquationState} and the additional assumption that the effective adiabatic index $\gamma(n)$ (defined in Eq.~(\ref{Eq:gamma(n)})) satisfies the bound $\gamma(n) \geq 4/3 + \varepsilon$ (with $\epsilon > 0$) for small enough values of the particle density $n$, we have provided a rigorous proof for the existence and uniqueness of global solutions of the TOV equations describing a static, spherical star of finite radius and mass. Furthermore, we have shown that the familiar Buchdahl bound $2m(r)/r < 8/9$ holds for any radius $r > 0$ (smaller than or equal to the radius of the surface of the star).

In particular, the results presented in this article apply to any perfect fluid with positive baryonic rest mass and a polytropic equation of state $p(n) = K n^\gamma$ with adiabatic index $\gamma > 4/3$. This includes the equation of state describing an ideal nonrelativistic monoatomic gas, for which $\gamma = 5/3$. Interestingly, the ultrarelativistic counterpart, for which $\gamma = 4/3$, is not included in our analysis. However, as discussed in detail in appendix~\ref{App:StatFis}, an ideal, relativistic monoatomic gas has an equation of state whose effective adiabatic index $\gamma(n)$ interpolates between the two values $4/3$ and $5/3$ in the limits $n\to \infty$ and $n\to 0$, respectively. Since our assumption on $\gamma(n)$ is only needed for small values of $n$ (and not in the ultrarelativistic limit $n\to \infty$), our results fully cover the case of the ideal relativistic monoatomic gas. It is only near the surface of the star (where $n$ is small and thus the gas is practically Newtonian) that the assumption $\gamma(n) \geq 4/3 + \varepsilon$ is required.

For a given equation of state fulfilling our assumptions, the quantitative properties of the star, like its radius, mass, density profile etc. can be obtained from numerical calculations. We have provided an example in section~\ref{Sec:Numerical} for a polytropic equation of state with adiabatic index $\gamma = 5/3$, although the method described in that section can be adapted to more general equations of state in a straightforward way. The most important feature found from the numerical calculations is the spiral-type behavior (see Fig.~\ref{Fig:MvsR}) in the mass-versus-radius relation for the resulting family of static, spherical stars and the existence of a maximum mass configuration in this family, which is important because it indicates a change in behavior for the stability of the star (see chapter 6 in Ref.~\cite{Shapiro}). Further numerical examples based on a dynamical system approach can be found in Ref.~\cite{Heinzle_2003}. For numerical time evolutions of (numerically perturbed) TOV stars, see for instance~\cite{fGfLmM12}.

Our proof for the global existence of stars with finite radius was mostly inspired by the work by Ramming and Rein~\cite{Ramming_2013} and the proof for the Buchdahl bound is a straightforward generalization of the arguments presented in section 6.2 in Ref.~\cite{Wald}. Although the results presented in this article are not new and have been widely studied in the literature, they are scattered in different articles and books. Therefore, we hope that our self-contained review regarding the most important results of the TOV equation and its solutions may serve as a useful pedagogical introduction to the topic and motivate research on  more realistic star models including rotation and magnetic fields, for which rigorous mathematical results are still scarce.


\acknowledgments

It is a pleasure to thank Emilio Tejeda and Thomas Zannias for useful comments on a previous version of this article and an anonymous referee for pointing out to us the relevant references concerning realistic equations of state for neutron stars. ECN was supported by the CONACYT project ``Ayudante de investigador"  No.~17840 and by a postgraduate CONACYT fellowship. OS was partially supported by a CIC grant to Universidad Michoacana de San Nicol\'as de Hidalgo.

\appendix
\section{Computation of the curvature and Einstein tensors}
\label{App:Curvature}

In this appendix, for completeness, we present details regarding the computation of the Riemann curvature, Ricci and Einstein tensors associated with an arbitrary, spherically symmetric metric of the form~(\ref{Eq:SphMetric}). The following presentation and notation follows the work in~\cite{eCnOoS13}. We assume a metric of the form
\begin{equation}
\label{Eq:metric}
g = \tilde{g} + r^2\hat{g},
\end{equation}
with $\tilde{g} = \tilde{g}_{ab} dx^a dx^b$ a two-dimensional Lorentzian metric and $\hat{g} = \hat{g}_{AB} dx^A dx^B =  d\vartheta^2 + \sin^2\vartheta d\varphi^2$ the standard metric on the two-sphere, and $r$ the radius function. For the static metric~(\ref{Eq:MetricAnsatz}) considered in the body of this article, the two-dimensional metric $\tilde{g}$ is of the form $\tilde{g} = -e^{2\Phi(r)} dt^2 + e^{2\Psi(r)} dr^2$ (see Eq.~(\ref{Eq:TwoMetric})) and its components only depend on the radius coordinate $r$. However, for the following calculations, nothing is lost by assuming a generic two-metric $\tilde{g}$ which depends on arbitrary coordinates $(x^a) = (x^0,x^1)$ and to consider the radius $r = r(x^0,x^1)$ to be a positive function of these coordinates. Such a generalization is useful, for instance, when considering time-dependent (non-static) spherically symmetric spacetimes or when discussing more general spacetimes in which $r$ cannot be used as a global coordinate (such as occurs in wormhole spacetimes, for instance).

Using the definition~(\ref{Eq:Christoffel}) for the Christoffel symbols, one finds
\begin{align}
\label{Eq:Chris1}
\Gamma^d{}_{ab} & = \tilde{\Gamma}^d{}_{ab}, \\
\Gamma^d{}_{aB} & = 0, \\
\label{Eq:Chris3}
\Gamma^D{}_{ab} & = 0, \\
\Gamma^d{}_{AB} & = -rr^d\hat{g}_{AB}, \\
\Gamma^D{}_{AB} & = \hat{\Gamma}^D{}_{AB}, \\
\Gamma^D{}_{aB} & = \frac{r_a}{r}\delta^D{}_B, 
\end{align}
where  $\tilde{\Gamma}^d{}_{ab}$ y $\hat{\Gamma}^D{}_{AB}$ are the Christoffel symbols associated with $\tilde{g}_{ab}$ and $\hat{g}_{AB}$, respectively, and where we recall that $a,b=0,1$ and $A,B = 2,3$ refer to the coordinates on the unit sphere. Also we introduced the notations $r_a := \partial_a r$ and $r^d := \tilde{g}^{da}r_a$. Now using these expressions and the formula~(\ref{Eq:Riemann}) for the Riemann curvature tensor, we obtain
\begin{align}
R^c{}_{dab} & = \tilde{R}^c{}_{dab}, \\
R^c{}_{Dab} & = 0, \\
R^C{}_{Dab} & = 0, \\
R^c{}_{DAB} & = 0, \\
R^c{}_{DaB} & = -r(\tilde{\nabla}^c\tilde{\nabla}_a r)\hat{g}_{BD}, \\
R^C{}_{DAB} & =  \hat{R}^C{}_{DAB} - r^er_e(\delta^C{}_A\hat{g}_{BD} - \delta^C{}_B\hat{g}_{AD}),
\end{align} 
where $\tilde{R}^c{}_{dab}$ and $\hat{R}^C{}_{DAB}$ refer to the components of the Riemann tensor associated with the metrics $\tilde{g}$ and $\hat{g}$ respectively. In two dimensions the curvature tensor has the following form (see, for instance exercise 4, chapter 3 in~\cite{Wald})
\begin{align}
\tilde{R}^c{}_{dab} & = \tilde{\kappa}(\delta^c{}_ {a} \tilde{g}_{bd} - \delta^c{}_{b} \tilde{g}_{ad}), \\
\hat{R}^C{}_{DAB} & = \hat{\kappa}(\delta^C{}_ {A} \hat{g}_{BD} - \delta^C{}_{B} \hat{g}_{AD}), 
\end{align}
where $\tilde{\kappa}$ and $\hat{\kappa}$ are the Gaussian curvatures associated with the metric $\tilde{g}$ and $\hat{g}$ respectively. Therefore,
\begin{align}
R^c{}_{dab} & = \tilde{\kappa}(\delta^c{}_ {a} \tilde{g}_{bd} - \delta^c{}_{b} \tilde{g}_{ad}), \\
R^c{}_{DaB} & = -r(\tilde{\nabla}^c\tilde{\nabla}_a r)\hat{g}_{BD}, \\
R^C{}_{DAB} & = (1 - r^e r_e)(\delta^C{}_ {A} \hat{g}_{BD} - \delta^C{}_{B} \hat{g}_{AD}).
\end{align}
With these expressions we can calculate the components of the Ricci tensor
\begin{align}
R_{ab} & = R^e{}_{aeb} + R^E{}_{aEb} = \tilde{\kappa}\tilde{g}_{ab} - \frac{2}{r}\tilde{\nabla}_a\tilde{\nabla}_b r, \\
R_{aB} & = R^e{}_{aeB} + R^E{}_{aEB} = 0, \\
R_{AB} & = (1 - r^er_e - r\tilde{\Delta}r)\hat{g}_{AB},
\end{align}
where $\tilde{\Delta}r = \tilde{\nabla}^b\tilde{\nabla}_b r = \tilde{g}^{ab}\tilde{\nabla}_a\tilde{\nabla}_b r$ is the covariant Laplacian of $r$. The Ricci scalar is given by 
\begin{equation}
R = R^a{}_{a} + R^A{}_{A} = 2\tilde{\kappa} + \frac{2}{r^2}(1 - r^er_e - 2r\tilde{\Delta}r).
\end{equation}
Finally, the components of the Einstein tensor are given by the following expressions 
\begin{align}
\label{Eq:TensorE11}
G^{a}{}_b & = - \frac{2}{r}\tilde{\nabla}^{a}{}\tilde{\nabla}_b r - \frac{1}{r^2}(1 - r^er_e - 2r\tilde{\Delta}r)\delta^{a}{}_{b}, \\
G^{a}{}_{B} & = 0, \\
\label{Eq:TensorE21}
G^{A}{}_{B} & = \left(\frac{\tilde{\Delta}r}{r} - \tilde{\kappa}\right)\delta^{A}{}_{B}.
\end{align}
Specializing to the case of the static two-metric~(\ref{Eq:TwoMetric}), one obtains from this the Christoffel symbols listed in Eqs.~(\ref{Eq:Christoffel1}--\ref{Eq:Christoffel4}) and the components of the Einstein tensor in Eqs.~(\ref{Eq:TE1},\ref{Eq:TE2},\ref{Eq:TE3}).

\section{Equation of state for a monoatomic, relativistic ideal gas}
\label{App:StatFis}

In this appendix we offer a derivation for the equation of state describing a classical (i.e. non-quantum) monoatomic, ideal gas, and towards the end of this appendix we also make some comments regarding the complete degenerate, ideal Fermi gas. To this purpose, we consider a fixed box of volume $V$ containing a large number $N$ of particles, but still assume that $V$ is small enough such that the metric is well-described (in a local inertial frame) by the Minkowski metric inside $V$, such that a special relativistic treatment inside $V$ is sufficient. We consider a system in which the temperature $T$ could be arbitrarily high, such that a significant fraction of the particles could have relativistic speeds, and thus we use the special relativistic Hamiltonian
\begin{equation}
H(x,p) = c\sum_{j = 1}^{N}\sqrt{|\vec{p}_j|^2 + m^2 c^2},
\end{equation}
with $p = (\vec{p}_1,\vec{p}_2,\ldots,\vec{p}_N)\in \Real^{3N}$ the momenta and $m$ the mass of the particles, to describe the system of $N$ particles. Based on these assumptions, we compute the thermodynamics of the gas using the canonical ensemble. The corresponding partition function is
\begin{equation}
Z(T, V, N) = \frac{1}{N! h^{3N}} \int e^{-\beta H(x,p)} d^{3N} x\, d^{3N} p,
\end{equation}
where $h$ is Planck's constant and $\beta = 1/(k_B T)$, $k_B$ denoting Boltzmann's constant. Since the gas is non-interacting, the partition function factorizes:
\begin{equation}
Z(T,V,N) = \frac{1}{N!} Z_1(T,V)^N,
\end{equation}
with
\begin{equation}
Z_1(T,V) = \frac{V}{h^3}\int e^{-c\beta\sqrt{|\vec{p}|^2 + m^2 c^2}} d^3 p.
\end{equation}
The integral can be computed using spherical coordinates, such that
\begin{equation}
Z_1(T,V) = \frac{4\pi V}{h^3}\int_0^\infty e^{-c\beta\sqrt{p^2 + m^2 c^2}} p^2 dp.
\end{equation}
Subsequently, one performs the variable substitution $p = mc\sinh\chi$ which yields
\begin{equation}
Z_1(T,V) = \frac{4\pi V}{\lambda^3}\int_0^\infty e^{-z\cosh\chi} \sinh^2\chi\cosh\chi d\chi,
\label{Eq:Z1}
\end{equation}
where we have introduced the Compton wavelength
\begin{equation}
\lambda := \frac{h}{m c}
\end{equation}
of the particles, as well as the dimensionless quantity
\begin{equation}
z := \beta m c^2 = \frac{m c^2}{k_B T},
\end{equation}
which is the ratio between the rest mass and thermal energy of the particles. Rewriting $\sinh^2\chi\cosh\chi = \frac{1}{3}\frac{d}{d\chi} \sinh^3\chi$ in Eq.~(\ref{Eq:Z1}) and using integration by parts leads to the final expression for the partition function:
\begin{equation}
Z(T, V, N) = \frac{1}{N!}\left[ \frac{4\pi V}{\lambda^3}\frac{K_2(z)}{z} \right]^N,
\end{equation}
where $K_2(z)$ denotes the modified Bessel function of the second kind of order $2$, see Ref.~\cite{DLMF} and Appendix~\ref{App:ModifiedBessel} for further details and its definition.

Using Stirling's approximation $\log N! = N\log N - N + {\cal O}(\log N)$, the free energy of the system is found to be
\begin{align*}
F(T,V,N) & = -k_B T\log Z(T, V, N)\\
& = -N k_B T\left\{ 1 + \log\left[ \frac{4\pi V}{\lambda^3 N}\frac{K_2(z)}{z} \right]
+ {\cal O}\left( \frac{\log N}{N} \right) \right\},
\end{align*}
from which one can easily compute the relevant thermodynamic quantities like pressure, entropy and internal energy using the well-known formulae (see, for instance~\cite{Huang}) 
\begin{equation}
p = -\left(\frac{\partial F}{\partial V}\right)_{T, N},\qquad
S = -\left(\frac{\partial F}{\partial T}\right)_{V, N},\qquad
U = F + T S.
\end{equation}
Dividing $S$ and $U$ by $V$ and taking the thermodynamic limit $N\to \infty$ holding the particle density $n := N/V$ constant, one obtains from this the following expressions for pressure, entropy density and energy density as functions of $(n,T)$:
\begin{eqnarray}
p(n,T) &=& n k_B T,
\label{Eq:FisStatp}\\
s(n,T) &=& n k_B \left\{ 4 + \log\left[ \frac{4\pi}{\lambda^3 n}\frac{K_2(z)}{z} \right]
 + z \frac{K_1(z)}{K_2(z)} \right\},
\label{Eq:FisStats}\\
\varepsilon(n,T) &=& n k_B T\left[ 3 +  z \frac{K_1(z)}{K_2(z)} \right].
\label{Eq:FisState}
\end{eqnarray}
In deriving these equations we have used the relation~(\ref{Eq:BesselRecurrence2}) to eliminate the derivative of $K_2$. Eq.~(\ref{Eq:FisStatp}) is the ideal gas equation, while from Eq.~(\ref{Eq:FisState}) we see that $\varepsilon(n,T)/n$ is a function of $T$ only which converges to the rest mass energy of the particles, $m c^2$, in the limit $T\to 0$ (see Eqs.~(\ref{Eq:K1Expansion},\ref{Eq:K2Expansion})). By construction, the first law~(\ref{Eq:FirstLaw}) is satisfied.

For an isentropic configuration, for which the specific entropy $s/n$ is constant, the second equation yields the following relation between $n$ and $T$:
\begin{equation}
n(T) = n_0\frac{K_2(z)}{z} e^{z\frac{K_1(z)}{K_2(z)}},\qquad
z := \frac{mc^2}{k_B T},
\label{Eq:nIsentropic}
\end{equation}
with $n_0$ a constant. The next lemma shows that this defines a smooth, strictly monotonously increasing function $n: (0,\infty) \to (0,\infty)$ which can hence be inverted to yield $T$ as a function of $n$. This allows one to eliminate the temperature in the expressions~(\ref{Eq:FisStatp},\ref{Eq:FisState}) and describe pressure and energy density as function of $n$ only. The formulae~(\ref{Eq:FisStatp},\ref{Eq:FisStats},\ref{Eq:FisState},\ref{Eq:nIsentropic}) were already derived over 100 years ago by F.~J\"uttner~\cite{fJ11}.

\begin{Lem}
The function $F: (0,\infty)\to (0,\infty)$  defined by
\begin{equation}
F(z) := \frac{K_2(z)}{z} e^{z\frac{K_1(z)}{K_2(z)}},\qquad z > 0,
\end{equation}
such that $n(T) = n_0 F(z)$, is smooth and satisfies $F'(z) < 0$ for all $z > 0$,  $z^3 F(z)\to 2$ for $z\to 0$ and $z^{3/2} F(z)\to \sqrt{\pi/(2e^3)}$ in the limit $z\to \infty$.
\end{Lem}

\begin{proof}
Differentiating the function $F$ and using the relations~(\ref{Eq:BesselRecurrence2}) yields
\begin{equation}
F'(z) = -G(z) \frac{K_2(z)}{z^2} e^{z\frac{K_1(z)}{K_2(z)}},\qquad
G(z) := z^2 + 3 - 3z\frac{K_1(z)}{K_2(z)} - z^2\frac{K_1(z)^2}{K_2(z)^2},
\label{Eq:GDef}
\end{equation}
and thus proving $F' < 0$ is equivalent to showing that $G(z) > 0$ for all $z > 0$. This in turn requires an upper bound for $K_1/K_2$. We first analyze the situation for small values of $z > 0$. In this case, one can use the estimate
\begin{equation}
\frac{K_1(z)}{K_2(z)} \leq z/2,\qquad z > 0,
\label{Eq:K1K2Est1}
\end{equation}
which follows from the recurrence relation~(\ref{Eq:BesselRecurrence1}) with $n=1$ and the fact that $K_0 > 0$. Using this into the definition of $G$ in Eq.~(\ref{Eq:GDef}) yields
\begin{equation}
G(z) \geq -\frac{z^2}{2} - \frac{z^4}{4} + 3 = \frac{1}{4}\left[ 13 - (z^2 + 1)^2 \right],
\end{equation}
which proves the $G(z) > 0$ for all $z^2 < \sqrt{13} - 1 \approx 2.6$.

To prove that $G$ is positive for larger values of $z$, we use instead the expansions~(\ref{Eq:K1Expansion},\ref{Eq:K2Expansion}) obtaining
\begin{equation}
\frac{K_1(z)}{K_2(z)} = 1  - \frac{3}{2z} + \frac{15}{8z^2}
\frac{1 + \frac{21}{32z} + \frac{8z^2}{15}[ r_{1,2}(z) - r_{2,2}(z)] + \frac{4z}{5} r_{2,2}(z)}
{1 + \frac{15}{8z} + \frac{105}{128z^2} + r_{2,2}(z)}.
\end{equation}
Using the estimates for the remainder terms $r_{1,2}$ and $r_{2,2}$ below Eqs.~(\ref{Eq:K1Expansion},\ref{Eq:K2Expansion}) yields the alternative estimate
\begin{equation}
\frac{K_1(z)}{K_2(z)} \leq 1  - \frac{3}{2z} + \frac{15}{8z^2},\qquad z > 0.
\label{Eq:K1K2Est2}
\end{equation}
which is better than~(\ref{Eq:K1K2Est1}) for large values of $z$. Combining this estimate with the definition of $G$ in Eq.~(\ref{Eq:GDef}) gives
\begin{equation}
G(z) \geq \frac{3}{2}\left( 1 - \frac{75}{32z^2} \right),
\end{equation}
which is positive for all $z^2 > 75/32 = 2.34375$. This proves that $G$ is positive and hence that $F'(z) < 0$ for all $z > 0$.

The claimed asymptotic behavior for $z\to 0$ and $z\to \infty$ follow easily from Eqs.~(\ref{Eq:ModifiedBesselExpansion},\ref{Eq:KnZeroExpansion}).
\end{proof}

It follows from the previous lemma that in the low temperature limit $z\to \infty$ (the symbol $\sim$ indicating proportionality)
\begin{equation}
n(T)\sim z^{-3/2} \sim T^{3/2},
\end{equation}
and thus $p\sim n^{5/3}$, whereas in the high temperature limit $z\to 0$ (i.e.  $k_B T\gg m c^2$),
\begin{equation}
n(T)\sim z^{-3} \sim T^3
\end{equation}
such that $p\sim n^{4/3}$. In particular, the assumptions $(i)$--$(iv)$ regarding the equation of state in section~\ref{sec:EquationState} are fulfilled and the effective adiabatic index defined in Eq.~(\ref{Eq:gamma(n)}) yields
\begin{equation}
\gamma(n) = 1 + \frac{1}{G(z)},
\label{Eq:IndiceAdiabatico}
\end{equation}
with $G$ defined in Eq.~(\ref{Eq:GDef}). From the asymptotic properties in the low temperature limit it follows that $\varepsilon/n \to m c^2 > 0$ and $\gamma(n) \to 5/3$ for $n\to 0$, and thus assumption $(iii)$ is satisfied for any $4/3 < \gamma_1 < 5/3$, which is sufficient to guarantee the existence of finite radius stars. A plot of the function $\gamma(n)$ is shown in Fig.~\ref{Fig:IndiceAdiabatico}.

\begin{figure}[ht]
\centerline{\includegraphics[width=11cm]{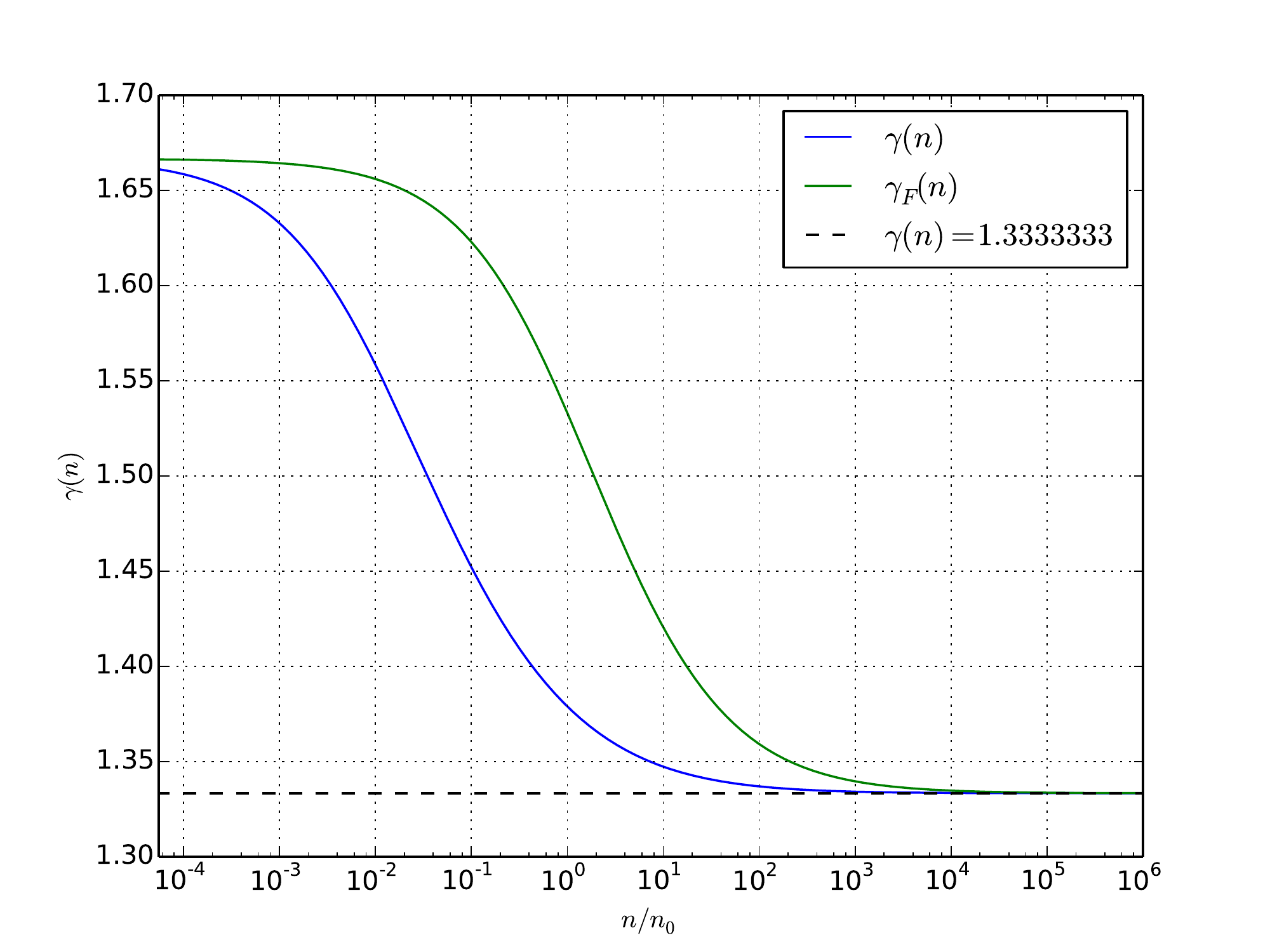}}
\caption{The effective adiabatic index as a function of $n$ for an ideal, relativistic monoatomic gas (blue curve) and for a completely degenerate Fermi gas (green curve). For the Fermi gas, we have defined $n_0 := (3\pi^2\lambda^3)^{-1}$.}
\label{Fig:IndiceAdiabatico}
\end{figure}

We end this appendix with a few remarks regarding the quantum analogue of the description we have given so far, assuming that the particles are fermions. For such a gas, the results we have discussed so far are only valid for high temperatures or low densities, such that $\lambda_T^3 n\ll 1$, with $\lambda_T = h/\sqrt{2\pi m k_B T}$ the thermal wavelength. For low temperatures and high densities, quantum mechanical effects need to be taken into account. This can be easily understood by noticing that the classical expressions for the pressure and energy density (see Eqs.~(\ref{Eq:FisStatp}) and (\ref{Eq:FisState})) converge to zero as $T\to 0$, while for a gas consisting of fermions these quantities cannot be zero due to Pauli's exclusion principle. A consistent generalization of the expressions~(\ref{Eq:FisStatp}--\ref{Eq:FisState}) for an ideal fermion gas should be based on quantum statistics (see for instance chapter 8 in~\cite{Huang} for the case of non-relativistic particles). Here, we only give the results for a completely degenerate Fermi gas, that is, a gas of fermions at zero temperature:
\begin{eqnarray}
n_F(x) &=& \frac{x^3}{3\pi^2\lambda^3},
\label{Eq:DegenFermin}\\
p_F(x) &=& \frac{m c^2}{8\pi^2\lambda^3}\left[ x\sqrt{1 + x^2}\left( \frac{2x^2}{3} - 1 \right)
 + \log\left( x + \sqrt{1 + x^2} \right) \right],
\label{Eq:DegenFermip}
\end{eqnarray}
where $x = \lambda k_F$ is the dimensionless Fermi momentum (see for instance chapter~2 in Ref.~\cite{Shapiro}). Eliminating $x$ from these expressions one obtains $p_F$ as a function of $n_F$, and the expression for $\varepsilon_F$ can be obtained by integrating the first law~(\ref{Eq:FirstLaw}) with $T = 0$ and setting $e_0 = m c^2$ in Eq.~(\ref{Eq:epsilonp}). Interestingly, the effective adiabatic index $\gamma_F(n)$ has the same qualitative properties as the one of the classical isentropic gas (see Eq.~(\ref{Eq:IndiceAdiabatico})), and interpolates between the values $5/3$ and $4/3$ as $n$ increases from $0$ to infinity. Its behaviour is also shown in Fig.~\ref{Fig:IndiceAdiabatico}. In particular, the assumptions $(i)$--$(iv)$ regarding the equation of state in section~\ref{sec:EquationState} are fulfilled.

The expressions~(\ref{Eq:DegenFermin},\ref{Eq:DegenFermip}) are relevant for the description of isolated white dwarfs and neutron stars, since these objects ultimately cool down to zero temperature.

\section{Definition and main properties of the modified Bessel functions of the second kind}
\label{App:ModifiedBessel}

In this appendix we briefly review the definition of the modified Bessel functions of the second kind $K_n$, $n=0,1,2,\ldots$, and some of their properties that are used in the previous appendix. We start with the following integral representation (see~\cite{DLMF}):
\begin{equation}
K_n(z) := 
\frac{z^n}{(2n-1)!!} \int_0^\infty e^{-z\cosh\chi} \sinh^{2n}\chi d\chi,\qquad z > 0,
\label{Eq:ModifiedBessel}
\end{equation}
where $(2n-1)!! = (2n-1)(2n-3)\cdots 3\cdot 1$. Using integration by parts and the identity $\cosh^2\chi - \sinh^2\chi = 1$, it is not difficult to prove the following recurrence relations:
\begin{eqnarray}
K_{n+1}(z) &=& \frac{2n}{z} K_n(z) + K_{n-1}(z),
\label{Eq:BesselRecurrence1}\\
K_n'(z) &=& \frac{n}{z} K_n(z) - K_{n+1}(z) = -\frac{n}{z} K_n(z) - K_{n-1}(z),
\label{Eq:BesselRecurrence2}
\end{eqnarray}
which are valid for all $n=1,2,3,\ldots$ and $z > 0$. Next, we are interested in the asymptotic behavior (with corresponding error estimates) for $z\to \infty$. For this, we first perform the variable substitution $\cosh\chi = 1 + \eta^2/(2z)$ in Eq.~(\ref{Eq:ModifiedBessel}), which yields
\begin{equation}
K_n(z) = \frac{1}{(2n-1)!!} \frac{e^{-z}}{\sqrt{z}} \int_0^\infty e^{-\frac{1}{2}\eta^2}
\eta^{2n}\left( 1 + \frac{\eta^2}{4z} \right)^{n-\frac{1}{2}} d\eta,\qquad z > 0.
\label{Eq:ModifiedBesselBis}
\end{equation}
Next, we use the Taylor expansion of the function $f(x) := (1 + x)^\alpha$ about the point $x = 0$:
\begin{equation}
f(x) = \sum\limits_{k=0}^N {\alpha\choose k} x^k 
 + {\alpha\choose N+1} (1 + \theta x)^{\alpha-N-1} x^{N+1},\qquad x > 0
\end{equation}
with some $0 < \theta < 1$ depending on $x$ and ${\alpha\choose k} = \alpha(\alpha-1)\cdots (\alpha-k+1)/k!$. Applying this to the integrand in Eq.~(\ref{Eq:ModifiedBesselBis}) with $x = \eta^2/(4z)$ and using the Gaussian integral
\begin{equation}
\int_0^\infty e^{-\frac{1}{2} \eta^2} \eta^{2n} d\eta = \sqrt{\frac{\pi}{2}} (2n-1)!!\; ,
\end{equation}
one obtains the following expansion:
\begin{equation}
K_n(z) = \sqrt{\frac{\pi}{2z}} e^{-z}\left[ \sum\limits_{k=0}^N \frac{a_k(n)}{z^k}
 + r_{n,N}(z) \right],\qquad z > 0,
 \label{Eq:ModifiedBesselExpansion}
\end{equation}
with the coefficients $a_0(n) := 1$,
\begin{equation}
a_k(n) = \frac{[4n^2-1][4n^2 - 9]\cdots [4n^2 - (2k-1)^2]}{8^k k!},\qquad k=1,2,3,\ldots,
\end{equation}
and the remainder term
\begin{equation}
r_{n,N}(z) = \frac{a_{N+1}(n)}{z^{N+1}} \frac{1}{(2n+2N+1)!!} \sqrt{\frac{2}{\pi}}
\int_0^\infty e^{-\frac{1}{2}\eta^2} \eta^{2n+2N+2}\left( 1 + \frac{\theta\eta^2}{4z} \right)^{n-N-\frac{3}{2}} d\eta,\qquad z > 0.
\label{Eq:ModifiedBesselRemainder}
\end{equation}
For $N >  n-1$ the exponent in the integrand on the right-hand side is negative and one obtains the estimate
\begin{equation}
0\leq \frac{r_{n,N}(z)}{a_{N+1}(n)} \leq \frac{1}{z^{N+1}},\qquad z > 0.
\label{Eq:ModifiedBesselErrorEstimate}
\end{equation}
Therefore, again provided that $N > n-1$, the remainder term has the same sign as the first neglected term $a_{N+1}(n)/z^{N+1}$ in the expansion~(\ref{Eq:ModifiedBesselExpansion}) and it is bounded by it in absolute value. The examples of relevance for the previous appendix are:
\begin{eqnarray}
K_1(z) &=& \sqrt{\frac{\pi}{2z}} e^{-z}
\left[ 1 + \frac{3}{8}\frac{1}{z} - \frac{15}{128}\frac{1}{z^2} + r_{1,2}(z) \right],
\label{Eq:K1Expansion}\\
K_2(z) &=& \sqrt{\frac{\pi}{2z}} e^{-z}\left[
1 + \frac{15}{8}\frac{1}{z} + \frac{105}{128}\frac{1}{z^2} + r_{2,2}(z) \right],
\label{Eq:K2Expansion}
\end{eqnarray}
with $0\leq z^3 r_{1,2}(z) \leq 105/1024$ and $-315/1024\leq z^3 r_{2,2}(z)\leq 0$.

Finally, we note that by pulling a factor $(4z)^{-n+\frac{1}{2}}$ out of the integral~(\ref{Eq:ModifiedBesselBis}) one can also understand the asymptotic limit of $K_n(z)$ for $z\to 0$. For example, one has
\begin{equation}
\lim\limits_{z\to 0} z^n K_n(z) = 2^{n-1}(n-1)!\, .
\label{Eq:KnZeroExpansion}
\end{equation}

\section{Completeness of the function space $X_R$}
\label{App:Completeness}

In this appendix we demonstrate that the set $X_R := C_b((0,R],\Real)$ of bounded, continuous, real-value functions on the interval $(0,R]$, equipped with the norm $\|\cdot\|_{\infty}$ defined in Eq.~(\ref{Eq:Norm}) forms a Banach space, that is, a complete normed space. For this, we first observe that $X_R$ is a real vector space. Next, we check that $\|\cdot\|_{\infty}$ satisfied the three postulates defining a norm, which are:

\begin{itemize}
\item[$1)$] $\|P\|_\infty \geq 0$ and $\|P\|_\infty = 0$ if and only if $P = 0$,
\item[$2)$] $\|\lambda P\|_\infty = |\lambda| \cdot \|P\|_\infty$ for all $\lambda \in \mathbb{R}$ and $P \in X_R$,
\item[$3)$] $\|P_1 + P_2\|_\infty \leq \|P_1\|_\infty + \|P_2\|_\infty$ for all $P_1, P_2 \in X_R$.
\end{itemize}
To this purpose, notice first that $|P(x)| \geq 0$ for all $x \in (0, R]$, hence it is clear that $\|P\|_\infty = \sup_{0 < x \leq R} |P(x)|\geq 0$ and that $\|P\|_\infty = 0$ if and only if $P(x) = 0$ for all $x \in (0,R]$. Hence, the first condition is satisfied. Next, we have
\begin{equation}
\|\lambda P\|_\infty = \sup_{0 < x \leq R} |\lambda P(x)| 
 = \sup_{0 < x \leq R}|\lambda| \cdotp |P(x)|
 = |\lambda|\sup_{0 < x \leq R} |P(x)| = |\lambda| \cdotp \|P\|_\infty,
\end{equation}
which shows that the second condition is also satisfied. Finally,
\begin{eqnarray}
\|P_1 + P_2\|_\infty &=& \sup_{0 < x \leq R} |P_1(x) + P_2(x)| \leq \sup_{0 < x \leq R} \left(|P_1(x)| + |P_2(x)|\right) 
\nonumber\\
&\leq& \sup_{0 < x \leq R} |P_1(x)| + \sup_{0 < x \leq R} |P_2(x)| = \|P_1\|_\infty + \|P_2\|_\infty,		
\end{eqnarray}
which shows that the third condition is also satisfied and leads to the conclusion that $\|\cdot\|_{\infty}$ defines a norm on $X_R$.

It remains to prove that $(X_R, \|\cdot\|_{\infty})$ is a Banach space. For this we must show that any Cauchy sequence $(P_k)$ converges in $(X_R, \|\cdot\|_{\infty})$, that is there exists a limit point $P\in X_R$ such that $\| P_k - P \|_\infty \to 0$ for $k\to \infty$. Therefore, let $(P_k)$ be a Cauchy sequence in $(X_R, \|\cdot\|_{\infty})$. This means that for any $\varepsilon > 0$ there exists $n \in \mathbb{N}$ such that  
\begin{equation}
\label{Eq:A1}
\sup_{0 < x \leq R} |P_k(x) - P_j(x)| = \|P_k - P_j\|_\infty < \varepsilon 
\end{equation} 
for all $k, j > n$. In particular 
\begin{equation}
\label{Eq:88}
|P_k(x) - P_j(x)| < \varepsilon
\end{equation}
for all $k, j > n$ and all $x \in (0, R]$. Thus $(P_k(x))$ is a Cauchy sequence in the complete space $(\Real, |\cdotp|)$, which implies that the limit 
\begin{equation}
\label{Eq:89}
P(x) := \lim_{k \to \infty}{P_k(x)}  \in \mathbb{R}
\end{equation}
exists for all $x\in (0,R]$. It remains to show that the function $P: (0,R]\to \Real$ defined in this way is continuous and bounded and that $P_k\to P$ in $(X_R, \|\cdot\|_{\infty})$.

\begin{Lem}
The function $P: (0, R] \rightarrow \mathbb{R}$ defined by Eq.~(\ref{Eq:89}) is continuous and bounded and $P_k \to P$ with respect to the norm $\|\cdotp\|_\infty$. 
\end{Lem}

\begin{proof}
Recall that continuity of $P$ at a point $x\in (0,R]$ means that if we take any sequence $(x_m)$ in $(0, R]$ which converges to $x \in (0, R]$, then we must have $P(x_m) \to P(x)$. Thus we need to prove that for all $\varepsilon > 0$ there exists a natural number $n_0 \in \mathbb{N}$ such that  
\begin{equation}
|P(x_m) - P(x)| < \varepsilon,
\end{equation}
for all $m > n_0$. Let $\varepsilon > 0$. Since $(P_k)$ is a Cauchy sequence, there exists $n_1 \in \mathbb{N}$ such that
\begin{equation}
|P_k(x) - P_j(x)| < \frac{\varepsilon}{3},
\end{equation}
for all $k, j > n_1$ and all $x \in (0, R]$. Taking the limit $k \to \infty$ on both sides of the inequality and taking the supremum over $x$, one obtains
\begin{equation}
\label{Eq:Sup}
\sup_{0 < x \leq R} |P(x) - P_j(x)| \leq \frac{\varepsilon}{3},
\end{equation}
for all $j > n_1$. Fix $j = n_1 + 1$. Due to the fact that $P_j$ is continuous, there exists $n_2 \in \mathbb{N}$ such that for all $m > n_2$,
\begin{equation}
|P_j(x_m) - P_j(x)| < \frac{\varepsilon}{3}.
\end{equation} 
Therefore, we find for all $m > n_2$,
\begin{eqnarray}
|P(x_m) - P(x)| &=& |P(x_m) - P_j(x_m) + P_j(x_m) - P_j(x) + P_j(x) - P(x)|
\nonumber\\
&\leq& |P(x_m) - P_j(x_m)| + |P_j(x_m) - P_j(x)| + |P_j(x) - P(x)|
\nonumber\\
&<& \frac{\varepsilon}{3} + \frac{\varepsilon}{3} + \frac{\varepsilon}{3} = \varepsilon.	
\end{eqnarray}
Thus, we conclude that $P$ is a continuous function. The inequality~(\ref{Eq:Sup}) implies that $P - P_j$ is bounded for all $j > n_1$, and hence $P = P - P_j + P_j$ is also bounded, implying that $P\in X_R$. Moreover, the same inequality~(\ref{Eq:Sup}) implies that $\|P - P_j\|_\infty < \varepsilon$ for all $j > n_1$, which shows that $P_j\to P$ in $(X_R,  \|\cdot\|_\infty)$. This concludes the proof of the lemma.
\end{proof}


\bibliographystyle{unsrt}
\bibliography{references}

\end{document}